\documentclass[envcountsame]{llncs}

\usepackage{latexsym,amsfonts}
\usepackage{amsmath}
\usepackage{amssymb}
\usepackage{array}
\usepackage{graphicx,color}
\usepackage{latexsym,amsfonts}
\usepackage{bbm}
\usepackage{tikz}
\usetikzlibrary{positioning}
\usepackage{url}

\newcommand{\entry}{\mathsf{entry}}
\newcommand{\vote}{\mathsf{vote}}
\newcommand{\cert}{\mathcal{Y}}
\newcommand{\minimize}[1]{\text{minimize} \ #1}
\newcommand{\maximize}[1]{\text{maximize} \ #1}
\newcommand{\minmarg}{\textsc{MinMargin}}

\DeclareMathOperator*{\argmax}{arg\,max}

\makeatletter
\newcommand\myitem[1][]{\item[#1]\def\@currentlabel{#1}}
\makeatother

\definecolor{blue}{RGB}{0,82,147} 
\definecolor{red}{RGB}{202,033,063}
\definecolor{green}{RGB}{98,158,31} 

\usepackage[breaklinks=true]{hyperref} 
\hypersetup{
	colorlinks=true,
	citecolor=green,
	linkcolor=blue,  
}

\newcounter{lpnumber} 
\newenvironment{linearprogram}[1]
{ \stepcounter{lpnumber}
  \begin{gather} #1 \tag{LP\arabic{lpnumber}} \\[-5ex] \notag
  \end{gather}
  \hspace{1cm} subject to \\[-8ex]
  \align }
{ \endalign }

\spnewtheorem{observation}[theorem]{Observation}{\bfseries}{\itshape}
\spnewtheorem{mythm}{Theorem}{\bfseries}{\itshape}
\spnewtheorem{mylem}{Lemma}{\bfseries}{\itshape}

\usepackage{environ}
\usepackage{mathdots}

\newcommand{\repeattheorem}[1]{%
  \begingroup
  \renewcommand{\themythm}{\ref{#1}}%
  \expandafter\expandafter\expandafter\mythm
  \csname reptheorem@#1\endcsname
  \endmythm
  \endgroup
}
\NewEnviron{reptheorem}[1]{%
  \global\expandafter\xdef\csname reptheorem@#1\endcsname{%
    \unexpanded\expandafter{\BODY}%
  }%
  \expandafter\theorem\BODY\unskip\label{#1}\endtheorem
}

\newcommand{\repeatlemma}[1]{%
  \begingroup
  \renewcommand{\themylem}{\ref{#1}}%
  \expandafter\expandafter\expandafter\mylem
  \csname replemma@#1\endcsname
  \endmylem
  \endgroup
}
\NewEnviron{replemma}[1]{%
  \global\expandafter\xdef\csname replemma@#1\endcsname{%
    \unexpanded\expandafter{\BODY}%
  }%
  \expandafter\lemma\BODY\unskip\label{#1}\endlemma
}

\begin{document}
\title{Popular Branchings and Their Dual Certificates\thanks{Part of this work was done at the 9th Eml\'ekt\'abla workshop in G\'ardony, Hungary.}}
\author{Telikepalli Kavitha\inst{1} \and Tam\'{a}s Kir\'{a}ly\inst{2} \and Jannik Matuschke\inst{3} \and \\ Ildik\'{o} Schlotter\inst{4} \and Ulrike Schmidt-Kraepelin\inst{5}}
\institute{TIFR, Mumbai, India; \email{kavitha@tifr.res.in} \and E\"otv\"os University, Budapest, Hungary; \email{tkiraly@cs.elte.hu} \and KU Leuven, Belgium; \email{jannik.matuschke@kuleuven.be} \and Budapest University of Technology and Economics, Hungary; \email{ildi@cs.bme.hu} \and Technische Universit\"at Berlin, Germany; \email{u.schmidt-kraepelin@tu-berlin.de}}
\maketitle
\pagestyle{plain}

\begin{abstract}
Let $G$ be a digraph where every node has preferences over its incoming edges. The preferences of a node extend
naturally to preferences over {\em branchings}, i.e., directed forests; a branching $B$ is  {\em popular} if $B$ does not lose a head-to-head
election (where nodes cast votes) against any branching. Such popular branchings have a natural application in liquid democracy. The popular branching problem is to decide if $G$ admits a popular branching or not.
We give a characterization of popular branchings in terms of {\em dual certificates} and use this characterization to design an efficient 
combinatorial algorithm for the popular branching problem. 
When preferences are weak rankings, we use our characterization to formulate the {\em popular branching polytope} in the original space and
also show that our algorithm can be modified to compute a branching with {\em least unpopularity margin}.
When preferences are strict rankings, we show that ``approximately popular'' branchings always exist. 
\end{abstract}

\section{Introduction}
\label{sec:intro}
Let $G$ be a directed graph where every node has preferences (in partial order) over its incoming edges. When $G$ is simple, the preferences can equivalently be defined on in-neighbors. 
We define a {\em branching} as a subgraph of $G$ that is a directed forest where any node has in-degree at most~1; a node with in-degree 0 is a \emph{root}.
The problem we consider here is to find a branching  that is {\em popular}. 

Given any pair of branchings, we say a node $u$ prefers the branching where it has a more preferred incoming edge
(being a root is $u$'s worst choice). If neither incoming edge is preferred to the other, then $u$ is indifferent between 
the two branchings.
So any pair of branchings, say $B$ and $B'$, can be compared by asking for the majority opinion, i.e., every node opts for the branching that it prefers,
and it abstains if it is indifferent between them.
Let $\phi(B,B')$ (resp., $\phi(B',B)$) be the number of nodes that prefer $B$ (resp., $B'$) in the $B$-vs-$B'$ comparison. 
If $\phi(B',B) > \phi(B,B')$, then we say $B'$ is {\em more popular} than $B$.

\begin{definition}
  \label{pop-branching}
  A branching $B$ is popular in $G$ if there is no branching that is more popular than $B$. That is, $\phi(B,B') \ge \phi(B',B)$ for all
branchings $B'$ in $G$.
\end{definition}

\noindent{\bf An application in computational social choice.}
We see the main application of popular branchings within \emph{liquid democracy}.
Suppose there is an election where a specific issue should be decided upon, and there are several proposed alternatives.
Every individual voter has an opinion on these alternatives, but might also consider certain other voters as being better informed than her.
Liquid democracy is a  novel voting scheme that
provides a middle ground between the feasibility of representative democracy and the idealistic appeal of direct democracy \cite{BlZu16a}: Voters can choose whether they delegate their vote to another, well-informed voter or cast their vote themselves. As the name suggests, voting power flows through the underlying network, or in other words, delegations are transitive.
During the last decade, this idea has been implemented within several online decision platforms such as {\em Sovereign} and {\em LiquidFeedback}\footnote{See www.democracy.earth and www.interaktive-demokratie.org, respectively.} and was used for internal decision making at Google \cite{HL15} and political parties, such as the German {\em Pirate Party} or the Swedish party {\em Demoex}.

In order to circumvent \emph{delegation cycles}, e.g., a situation in which voter $x$ delegates to voter $y$ and vice versa, and to enhance the expressiveness of delegation preferences, several authors proposed to let voters declare a set of acceptable representatives \cite{GKMP18} together with a preference relation among them \cite{Bril18a,HL15,KoRi18a}. Then, a mechanism selects one of the approved representatives for each voter, avoiding delegation cycles. Similarly as suggested in \cite{ChGr17a}, we additionally assume that voters accept themselves as their least preferred approved representative. 

This reveals the connection to branchings in simple graphs (with loops), where nodes correspond to voters and the edge $(x,y)$ indicates that voter $x$ is an approved representative of voter $y$.\footnote{Typically, such a delegation is represented by an edge $(y,x)$; for the sake of consistency with downward edges in a branching, we use $(x,y)$.} 
{Every root in the branching casts a weighted vote on behalf of all her descendants.}
{What is a good mechanism to select representatives for voters?}
A crucial aspect in liquid democracy is the {\em stability} of the delegation process~\cite{BGL18a,EGP19a}. For the {model described above}, 
we propose popular branchings as a new concept of stability, i.e., {the majority of the electorate will always weakly prefer to delegate votes along the edges of a popular branching as opposed to delegating along the edges of any other branching.}

Not every directed graph admits a popular branching. Consider the following 
{simple graph} on four nodes $a, b, c, d$ where $a, b$ (similarly, $c, d$) are each other's top choices, while $a, c$ (similarly, $b, d$) are each other's second choices. There is no edge between $a,d$ (similarly, $b,c$). 
Consider the branching $B = \{(a,b),(a,c),(c,d)\}$.
A more popular branching is $B' = \{(d,c),(c,a),(a,b)\}$. Observe that $a$ and $c$ prefer $B'$ to $B$, while $d$ prefers $B$ to $B'$ 
and $b$ is indifferent between $B$ and $B'$. We can similarly obtain a branching $B'' = \{(b,a),(b,d),(d,c)\}$ that is more popular than $B'$. 
It is easy to check that this instance has no popular branching.

\subsection{Our Problem and Results}
The popular branching problem is to decide if a given digraph $G$ admits a popular branching or not, and if so, to find one. 
We show that determining whether a given branching $B$ is popular is equivalent to solving a min-cost arborescence problem in an extension of $G$ with appropriately defined edge costs (these edge costs are a {\em function} of the arborescence).
%weights.
The dual LP to this arborescence problem gives rise to a laminar set system that serves as a certificate for the popularity of $B$ if it is popular.
This dual certificate proves crucial in devising an algorithm for efficiently solving the popular branching problem.

\begin{theorem}
  \label{thm:main}
  Given a directed graph $G$ where every node has preferences in arbitrary partial order over its incoming edges, there is
  a polynomial-time algorithm to decide if $G$ admits a popular branching or not, and if so, to find one.
\end{theorem}

The proof of Theorem~\ref{thm:main} is presented in Section \ref{sec:algo}; it is based on a characterization of popular branchings that we develop in Section \ref{sec:lp}. 
In applications like liquid democracy, it is natural to assume that the preference order of every node is a \emph{weak ranking}, i.e., a ranking of its incoming edges with possible ties.
In this case, the proof of correctness of our popular branching algorithm leads to a formulation of the {\em popular branching polytope} ${\cal B}_G$, i.e., 
the convex hull of incidence vectors of popular branchings in $G$.  

\begin{theorem}
  \label{thm:polytope}
  Let $G$ be a digraph on $n$ nodes and $m$ edges where every node has a weak ranking over its incoming edges.
  The popular branching polytope of $G$ admits a formulation of size $O(2^n)$ in $\mathbb{R}^m$.
  Moreover, this polytope has $\Omega(2^n)$ facets.
\end{theorem}

We also show an extended formulation of ${\cal B}_G$ in $\mathbb{R}^{m+mn}$ with $O(mn)$ constraints.
When $G$ has edge costs and node preferences are weak rankings, the min-cost popular branching problem can be efficiently solved. 
So we can efficiently solve extensions of the popular branching problem, such as finding one that minimizes the largest rank 
used or one with given forced/forbidden edges. 

\medskip

\noindent{\bf Relaxing popularity.}
Since popular branchings need not always 
exist in $G$, this motivates relaxing popularity to {\em approximate popularity}---do approximately popular branchings always exist in any
instance $G$? An approximately popular branching~$B$ may lose an election against another branching, however the extent of this defeat will
be bounded. There are two measures of unpopularity: {\em unpopularity factor} $u(\cdot)$ and {\em unpopularity margin} $\mu(\cdot)$.
These are defined as follows:
\[ u(B) = \max_{\phi(B',B)>0} \frac{\phi(B',B)}{\phi(B,B')} \ \ \ \text{and}\ \ \ \mu(B) = \max_{B'}\, \phi(B',B) - \phi(B,B').\]

A branching $B$ is popular if and only if $u(B) \le 1$ or $\mu(B) = 0$. 
We show the following results.

\begin{theorem}[$\star$\footnote{Theorems marked by an asterisk ($\star$) are proved in the \nameref{ch:appendix}.
}]
  \label{thm:unpop-margin}
  A branching with minimum unpopularity margin in a digraph where every node has a weak ranking over its incoming edges can be efficiently computed. In contrast, when node preferences are in arbitrary partial order, the minimum unpopularity margin problem is $\mathsf{NP}$-hard. 
\end{theorem}

\begin{theorem}[$\star$]
  \label{thm:unpop-factor}
  Let $G$ be a digraph where every node has {a strict ranking} over its incoming edges. Then there always exists a branching $B$ in $G$ with
  $u(B) \le \lfloor\log n\rfloor$. Moreover, for every $n$, we can show an instance $G_n$ on $n$ nodes with strict 
  rankings such that $u(B) \ge \lfloor\log n\rfloor$ for every branching $B$ in $G_n$.
\end{theorem}

\medskip

\noindent{\bf Hardness results for restricted popular branching problems.}
A natural optimization problem here is to compute a popular branching where no tree is large. In liquid democracy, a large-sized tree shows a high
concentration of power in the hands of a single voter, and this is harmful for social welfare~\cite{GKMP18}. 
When there is a fixed subset of root nodes in a directed graph, it was shown
in \cite{GKMP18} that it is $\mathsf{NP}$-hard to find a branching that minimizes the size of the largest tree.
To translate this result to popular branchings,
we need to allow ties, whereas Theorem~\ref{thm:hardness_load} below holds even for strict {rankings}.
Another natural restriction is to limit the out-degree of nodes; 
Theorem~\ref{thm:hardness_load} also shows that this variant is computationally hard.

\begin{reptheorem}{restate_thm:hardness_load}[$\star$]
  \label{thm:hardness_load}
  Given a digraph $G$ where each node has {a strict ranking} over its incoming edges, it is $\mathsf{NP}$-hard to decide if there exists 
\begin{itemize}\vspace{-10pt}
  \item[(a)] a popular branching in $G$ where each node has at most $9$ descendants;
  \item[(b)] a popular branching in $G$ with maximum out-degree at most $2$.
\end{itemize}    
\end{reptheorem}

\subsection{Background and Related Work}
The notion of popularity was introduced by G\"ardenfors~\cite{Gar75} in 1975 in the domain of bipartite matchings. Algorithmic questions in popular matchings have been well-studied for the last 10-15 years~\cite{AIKM07,BIM10,CHK15,CK16,FK20,FKPZ19,GMSZ19,HK11,HK17,Kav14,Kav16,KMN09,McC06}. 

Algorithms for popular matchings were first studied in the {\em one-sided} preferences model where vertices on only one side of the bipartite graph have preferences over their neighbors. Popular matchings need not always exist here and there is an efficient algorithm to solve the popular matching problem~\cite{AIKM07}. The functions unpopularity factor/margin were introduced in \cite{McC06} to measure the {\em unpopularity} of a matching; it was shown in \cite{McC06} that it is $\mathsf{NP}$-hard to compute a matching that minimizes either of these quantities.
In the domain of bipartite matchings with {\em two-sided} strict preferences, popular matchings always exist since stable matchings always exist~\cite{GS62} and every stable matching is popular~\cite{Gar75}.

 The concept of popularity has previously been applied to (undirected) spanning trees \cite{Darm13b,Darm16,DKP11a}. In contrast to our setting, voters have rankings over the entire edge set. This allows for a number of different ways to derive preferences over trees, most of which lead to hardness results. 

\medskip

\noindent{\bf Techniques.} We characterize popular branchings in terms of {\em dual certificates}.  This is analogous to characterizing popular matchings in terms of {\em witnesses} (see \cite{FK20,HK17,Kav16}). However, witnesses of popular matchings are in $\mathbb{R}^n$ and these 
are far simpler than  
dual certificates. A dual certificate is an appropriate family of {subsets of}
the node set $V$.
A certificate of size $k$ implies that the unpopularity margin of the branching is at most $n-k$.
Our algorithm constructs a partition ${\cal X}'$ of $V$ such that if $G$ admits popular branchings, then there has to be 
{\em some} popular branching in $G$ with a dual certificate of {size $n$   supported by ${\cal X}'$}.
Moreover, when nodes have weak rankings, ${\cal X}'$ {supports some dual certificate of size $n$}
to {\em every} popular branching in $G$: this leads to the formulation of ${\cal B}_G$ (see Section~\ref{sec:polytope}). 
Our positive results on low unpopularity branchings are extensions of our algorithm.

\medskip

\noindent{\bf Notation.} The preferences of node $v$ on its incoming edges are given by a strict partial order $\prec_v$, so $e \prec_v f$ means that $v$ prefers edge $f$ to edge $e$. We use $e \sim_v f$ to denote that $v$ is indifferent between $e$ and $f$, that is, neither $e \prec_v f$ nor $e \succ_v f$ holds. 
The relation $\prec_v$ is a \emph{weak ranking} if $\sim_v$ is transitive. In this case, $\sim_v$ is an equivalence relation and there is a strict order on the equivalence classes. {When each equivalence class has size~1, we call it a \emph{strict ranking}.}

\section{Dual Certificates}
\label{sec:lp}

We add a dummy node $r$ to $G = (V_G,E_G)$ as the {\em root} and make {$(r,v)$ the least preferred incoming edge of any node} 
{$v$} in $G$.  
Let $D = (V \cup \{r\}, E)$ be the resulting graph where $V = V_G$ and $E = E_G \cup\{(r,u): u \in V\}$. 
An \emph{$r$-arborescence} {in $D$} is an out-tree with root $r$ (throughout the paper, all arborescences are assumed to be rooted at $r$ and to span %$V_G$, 
{$V$}, unless otherwise stated).

Note that there is a one-to-one correspondence between branchings in $G$ and arborescences in $D$ (simply make $r$ the parent of all roots of the branching). 
A branching is popular in $G$ if and only if the corresponding arborescence is popular among all arborescences in~$D$.\footnote{Note that, by the special structure of $D$, this is equivalent to $A$ being a popular branching in $D$.} 
We will therefore prove our results for arborescences in~$D$. The corresponding results for %popular 
branchings in $G$ follow immediately by projection, i.e., removing node $r$ and its incident edges.

Let $A$ be an arborescence in $D$. There is a simple way to check if $A$ is popular in $D$.
Let $A(v)$ be the incoming edge of $v$ in $A$. For $e=(u,v)$ in $D$, define:
 \[c_A(e) := \begin{cases} 0, \quad \text{if } e \succ_v A(v), \ \ \text{i.e.,\ $v$\ prefers\ $e$\ to\ $A(v)$;}\\  
 1,\quad \text{if } e \sim_v A(v),  \ \ \text{i.e.,\ $v$\ is\ indifferent\ between\ $e$\ and\ $A(v)$;}\\ 
 2, \quad \text{if } e \prec_{v} A(v), \ \ \text{i.e.,\ $v$\ prefers\ $A(v)$\ to\ $e$.} \end{cases}\] 

Observe  that $c_A(A) = |V| = n$ since $c_A(e) = 1$ for every $e \in A$.
 Let $A'$ be any arborescence in $D$ and let $\Delta(A,A') = \phi(A,A') - \phi(A',A)$ be the difference in the number of votes for $A$ and the number of votes for $A'$ in the $A$-vs-$A'$ {comparison}.
 Observe that $c_A(A') = \Delta(A,A') + n$. Thus, $c_A(A') \ge n = c_A(A)$ if and only if $\Delta(A,A') \ge 0$. So we can conclude the following.

 \begin{proposition}
   \label{prop0}
    Let $A'$ be a min-cost arobrescence in $D$ with respect to~$c_A$. Then $\mu(A) = n - c_A(A')$. In particular, $A$ is popular in $D$ if and only if it is a min-cost arborescence in $D$ with edge costs given by $c_A$.
 \end{proposition}

Consider the following linear program \ref{LP1}, which computes a min-cost arborescence in $D$, and its dual \ref{LP2}. For any non-empty $X \subseteq V$,
 let $\delta^-(X)$ be the set of edges entering the set $X$ in the graph $D$.
\begin{linearprogram}
  {
    \label{LP1}
    \minimize \sum_{e \in E} c_A(e)\cdot x_e  
  }
  \qquad\sum_{e \in {\delta}^-(X)}x_e \ & \ge \ \ 1  \ \ \, \forall\, X \subseteq V, \ X \ne \emptyset \notag\\
  x_e  \ & \ge \ \ 0   \ \ \ \forall\, e \in E. \notag
\end{linearprogram}
\vspace{-0.6cm}
\begin{linearprogram}
  {
    \label{LP2}
    \maximize \sum_{X \subseteq V,\, X\ne\emptyset}y_X
  }
  \sum_{X:\, \delta^-(X) \ni e}y_X \ & \le \ \ c_A(e) \ \ \ \, \forall\, e\in E \notag\\
  y_X \ & \ge \ \ 0 \ \ \ \ \ \ \ \ \ \,\forall\, X \subseteq V, \ X \ne \emptyset. \notag
\end{linearprogram}

For any feasible solution $y$ to \ref{LP2}, let ${\cal F}_y := \{X \subseteq V: y_X > 0\}$ be the support of $y$. Inspired by Edmonds' branching algorithm~\cite{Edmo67a}, Fulkerson~\cite{Fulk74a} gave an algorithm to find an optimal solution $y$ to \ref{LP2} such that $y$ is integral. From an alternative proof in \cite{KoVy06a}, we obtain the following useful lemma.

\begin{lemma}
  \label{lem:dualLaminarity}
There exists an optimal, integral solution $y^*$ to \ref{LP2} such that $\mathcal{F}_{y^*}$ is laminar. 
\end{lemma}

  Let $y$ be an optimal, integral solution to \ref{LP2} such that $\mathcal{F}_{y}$ is laminar. Note that for any nonempty $X \subseteq V$, there is an $e \in A \cap \delta^{-}(X)$ and thus $y_X \leq c_A(e) = 1$. This implies that $y_X \in \{0, 1\}$ for all $X$. We conclude that $\mathcal{F}_y$ is a dual certificate for $A$ in the sense of the following definition.
  
\begin{definition}
  \label{def:dual-certificate}
A dual certificate for $A$ is a laminar family $\mathcal{Y} \subseteq 2^V$ such that $|\{X \in \mathcal{Y} : e \in \delta^-(X)\}| \leq c_A(e)$ for all $e \in E$.
\end{definition}

For the remainder of this section, let $\mathcal{Y}$ be a dual certificate maximizing $|\mathcal{Y}|$.

\begin{lemma}\label{lem:one-to-one-sets-and-edges}\label{lem:dual-characterization}
  Arborescence $A$ has unpopularity margin $\mu(A) = n - |\mathcal{Y}|$. Furthermore, the following three statements are equivalent:\vspace{-0.2cm}
  \begin{itemize}
    \item[(1)] $A$ is popular.
    \item[(2)] $|\mathcal{Y}| = n$.
    \item[(3)] $|A \cap \delta^{-}(X)| = 1$ for all $X \in \mathcal{Y}$ and $|\{X \in \mathcal{Y} : e \in \delta^-(X)\}| = 1$ for all $e \in A$.
  \end{itemize}
\end{lemma}
\begin{proof}
  Let $x$ and $y$ be the characteristic vectors of $A$ and $\mathcal{Y}$, respectively.
  By Proposition~\ref{prop0}, $A$ is popular if and only if $x$ is an optimal solution to \ref{LP1}. This is equivalent to (2) because $c_A(A) = n$. Note also that (3) is equivalent to $x$ and $y$ fulfilling complementary slackness, which is equivalent to $x$ being optimal. \qed
\end{proof}

Lemma~\ref{lem:one-to-one-sets-and-edges} establishes the following one-to-one correspondence between the nodes in $V$ and the sets of $\mathcal{Y}$: For every set $X \in \mathcal{Y}$, there is a unique edge $(u, v) \in A$ that enters $X$. We call $v$ the \emph{entry-point} for $X$. Conversely, we let $Y_v$ be the unique set in $\mathcal{Y}$ for which $v$ is the entry-point; thus $\mathcal{Y} = \{Y_v : v \in V\}$.

\begin{observation}
  \label{obs:rootedTwoLayered}
  For every $v \in V$ we have $|\{X \in \mathcal{Y} : v \in X\}| \leq 2$.
\end{observation} 

Observation~\ref{obs:rootedTwoLayered} is implied by the fact that $e = (r,v)$ is an edge in $D$ for every $v \in V$ and $c_A(e) \leq 2$. Laminarity of $\mathcal{Y}$ yields the following corollary:

\begin{corollary}\label{cor:singletons}
If $|\mathcal{Y}| = n$, then $w \in Y_v \setminus \{v\}$ for some $v \in V$ implies $Y_w = \{w\}$.
\end{corollary}

The following definition of the set of \emph{safe} edges $S(X)$ with respect to a subset $X \subseteq V$ will be useful. $S(X)$ is the set of edges $(u,v)$ in $E[X]: = E \cap (X \times X)$ such that properties~\ref{property-1} and \ref{property-2} hold:
\begin{enumerate}
\item\label{property-1} $(u,v)$ is {\em undominated} in $E[X]$, i.e., $(u,v) \not\prec_v (u',v) \ \forall\,(u',v) \in E[X]$.
\item\label{property-2} $(u,v)$ {\em dominates} $(w,v)$ {$\forall\, w \notin X$}, i.e., $(u,v) \succ_v (w,v) \ \forall\, (w,v) \in \delta^-(X)$.
\end{enumerate}

The interpretation of $S(X)$ is the following. 
Suppose that the dual certificate~$\mathcal{Y}$ proves the popularity of $A$. Let $X \in \mathcal{Y}$ with $|X| > 1$.
By Corollary~\ref{cor:singletons}, for every node $v \in X$ other than the entry-point in $X$ we have $\{v\} = Y_v \in \mathcal{Y}$. So edges in $\delta^-(v)$ within $E[X]$ enter exactly one dual set, i.e., $\{v\}$, while any edge $(w,v)$ where $w \notin X$ enters two of the dual sets: $X$ and $\{v\}$. 
This induces exactly the constraints (\ref{property-1}) and (\ref{property-2}) given above on $(u,v) \in A$ (see \ref{LP2}), showing that the edge $A(v)$ must be safe, 
as stated in Observation~\ref{obs:safe-edges}.

\begin{observation}
\label{obs:safe-edges}
  If $A$ is popular, then $A \cap E[X] \subseteq S(X)$ for all $X \in \mathcal{Y}$.
\end{observation}

\section{Popular Branching Algorithm}
\label{sec:algo}

We are now ready to present our algorithm for deciding if $D$ admits a popular arborescence or not. 
For each $v \in V$, step~\ref{step-1} builds the largest set $X_v$ such that $v$ can reach all nodes in $X_v$ using edges in $S(X_v)$. The collection ${\cal X} = \{X_v: v \in V\}$ will be laminar (see Lemma~\ref{lem:u-in-W_v}).
To construct the sets $X_v$ we make use of the {\em monotonicity} of $S$: $X \subseteq X'$ implies $S(X) \subseteq S(X')$. 

In steps~\ref{step-2}-\ref{step-3}, the algorithm contracts each maximal set in ${\cal X}$ into a single node and builds a graph $D'$ on these
nodes and $r$. For each set $X$ here that has been contracted into a node, edges incident to $X$ in $D'$ are undominated edges
from other nodes in $D'$ to the {\em candidate entry-points} of $X$, which are nodes $v \in X$ such that $X = X_v$.
Our proof of correctness (see Theorems~\ref{thm:alg-rootedpop-safe}-\ref{thm:no-return-no-poparb}) shows that $D$ admits a
popular arborescence if and only if $D'$ admits an arborescence.

Our algorithm for computing a popular arborescence in $D$ is given below.

\begin{enumerate}
\item\label{step-1} For each $v \in V$ do:
\begin{itemize}
  \item let $X_v^0 = V$ and $i=0$;
  \item while $v$ does not reach all nodes in the graph $D^i_v = (X^i_v, S(X^i_v))$ do:
    \smallskip
    \begin{enumerate}
        \item[] $X^{i+1}_v = $ the set of nodes reachable from $v$ in $D^i_v$; let $i = i+1$.
    \end{enumerate}  
  \item let $X_v = X^i_v$.
\end{itemize}    

\smallskip

\item\label{step-2} Let $\mathcal{X} = \{X_v : v \in V\}$, $\mathcal{X}' = \{X_v\in \mathcal{X} : X_v \text{ is} \subseteq\text{-maximal in }\mathcal{X}\}$, $E' = \emptyset$.

  \smallskip

\item\label{step-3} For every edge $e = (u,v)$ in $D$ such that $X_v \in \mathcal{X}'$ and $u \notin X_v$ do:
  \begin{itemize}
 \item {if $e$ is undominated (i.e.,  $e \not\prec_v e'$) among all edges $e' \in \delta^-(X_v)$,} then
  	\vspace{-6pt}
    \[ \vspace{-6pt}
    f(e) = \begin{cases} (U,X_v) \quad \text{where } u \in U \ \text{and}\ U \in {\mathcal{X}'},\\
      (r,X_v) \quad \text{ if } u = r;\end{cases}\] 
   \item let $E' := E' \cup \{f(e)\}$.
  \end{itemize}

\smallskip
  
\item\label{step-4} If $D' = (\mathcal{X}'\cup\{r\}, E')$ contains an arborescence $\tilde{A}$, then
  \begin{itemize}
    \item let $A' = \{e: f(e) \in \tilde{A}\}$; 
    \item let $R = \{v \in V: \, |X_v| \ge 2$ and $v$ has an incoming edge in $A'\}$;
    \item for each $v \in R$: let $A_v$ be an arborescence in $(X_v,S(X_v))$;
    \item return $A^* = A' \cup_{v \in R}A_v$.
  \end{itemize}

\smallskip
  
\item\label{step-5} Else return ``{\em No popular arborescence in $D$}''.  
\end{enumerate}  

\paragraph{\bf Correctness of the above algorithm.}
We will first show the easy direction, that is, if the algorithm returns an
edge set $A^*$, then $A^*$ is a popular arborescence in $D$. The following lemma will be key to this. 
Note that the set $X_u$, for each $u \in V$, is defined in step~\ref{step-1}.
Lemmas marked by ($\circ$) are proved in the \nameref{ch:appendix}.

\begin{replemma}{restate_lem:u-in-W_v}[$\circ$]
\label{lem:u-in-W_v}
 $\mathcal{X} = \{X_v: v \in V\}$ is laminar. If $u \in X_{v}$, then $X_u \subseteq X_v$.
\end{replemma}

\begin{reptheorem}{restate_thm:alg-rootedpop-safe}[$\star$]
\label{thm:alg-rootedpop-safe}
If the above algorithm returns an edge set $A^*$, then $A^*$ is a popular arborescence in $D$.  
\end{reptheorem}
\begin{proof}[Sketch]
  It is straightforward to verify that $A^*$ is an arborescence in $D$. 
  To prove the popularity of $A^*$, we construct a dual certificate $\mathcal{Y}$ of %value 
  size $n$ for $A^*$, 
  by setting $\mathcal{Y}:= \{X_v: v \in R\} \cup \{\{v\}: v \in V \setminus R\}$.
  
  Note that $|\mathcal{Y}| = |R| + |V\setminus R| = n$. 
  It remains to show that any edge $(w,v) \in E$ satisfies the constraints in \ref{LP2}; 
  let $(u,v)$ be the incoming edge of $v$ in~$A^*$.

Suppose $v \in R$; then $(u,v) \in A'$ and $u \notin X_v$. 
{Consider any edge} $(w,v)$: this enters one set of $\mathcal{Y}$ iff $w \not\in X_v$ and no set iff $w \in X_v$. 
Hence, it suffices to show that $c_{A^*}((w,v)) \in \{1,2\}$ for $w \notin X_v$. 
By construction of ${E}'$, $(w,v)$ does not dominate $(u,v)$ and therefore $c_{A^*}((w,v)) \in \{1,2\}$. 

Suppose $v \in V \setminus R$. Let $s$ be $v$'s local root, i.e., the unique $s \in R$ with $v \in X_s$. Then $(u,v) \in A_s \subseteq S(X_s)$ by construction of $A_s$.
Any edge $(w,v) \in \delta^-(v)$ enters at most two sets of $\mathcal{Y}$:
$\{v\}$ and possibly $X_s$.
If, on the one hand, $(w, v) \in \delta^-(X_s)$, then $(u,v) \in S(X_s)$ dominates $(w,v)$ by property~\ref{property-2} of $S(X_s)$, and hence $c_{A^*}((w,v)) = 2$. 
If, on the other hand, $w \in X_s$, then 
$(u,v) \in S(X_s)$ is not dominated by $(w,v)$ by property~\ref{property-1} of $S(X_s)$, 
and hence $c_{A^*}((w,v)) \geq 1$. 
Thus, any edge satisfies the constraints in \ref{LP2}, proving the theorem.
\qed
\end{proof}

\begin{theorem}
\label{thm:no-return-no-poparb}
If $D$ admits a popular arborescence, then our algorithm finds one.
\end{theorem}

Before we prove Theorem~\ref{thm:no-return-no-poparb}, we need Lemma~\ref{lem:witnessInWv} and Lemma~\ref{lem:poparb-enters-W-once}. 

\begin{replemma}{restate_lem:witnessInWv}[$\circ$]
  \label{lem:witnessInWv}
Let $A$ be a popular arborescence and $\mathcal{Y}$ a dual certificate {for $A$ of size $n$}.
Then $Y_v \subseteq X_v$ for any $v \in V$.
\end{replemma}

\begin{replemma}{restate_lem:poparb-enters-W-once}[$\circ$]
\label{lem:poparb-enters-W-once}
Let $A$ be a popular arborescence in $D$ and let $X \in \mathcal{X}'$. Then $A$ enters $X$ exactly once, and it enters $X$ at some node $v$ such that $X = X_v$.
\end{replemma}

\noindent{\bf Proof of Theorem~\ref{thm:no-return-no-poparb}.}
Assume there exists a popular arborescence $A$ in $D$; then there exists a dual certificate $\mathcal{Y}$ 
{of size $n$} for $A$.
We will show there exists an arborescence in $D'$. 
By Lemma~\ref{lem:poparb-enters-W-once}, for each $X \in {\mathcal{X}}'$ there exists exactly one edge $e_X = (u,v)$ of
$A$ that enters $X$, and moreover, $v$ is a candidate entry-point of~$X$. 

We claim that $(u,v)$ is not dominated by any
$(u',v) \in \delta^-(X)$. 
Recall that by Lemma~\ref{lem:witnessInWv}, we know $Y_v \subseteq X_v = X$. If some $(u',v) \in \delta^-(X)$ dominates $(u,v) \in A$, its cost must be
$c_{A}((u',v)) =0$.
However, $(u',v)$ clearly enters $Y_v \subseteq X$, and thus violates \ref{LP2}, contradicting our assumption
that $\mathcal{Y}$ is a dual solution.
Hence, $e_X$ is undominated among the edges of $\delta^-(X) \cap \delta^-(v)$ and therefore our algorithm creates an edge
$f(e_X)$ in $E'$ pointing to $X$. Using the fact that $A$ is an arborescence in $D$, it is straightforward to verify that
the edges $\{f(e_X): X \in {\cal X}'\}$ form an arborescence $\tilde{A}$ in $D'$. 
Thus our algorithm returns an edge set $A^*$,
which by Theorem~\ref{thm:alg-rootedpop-safe} must be a popular arborescence in $D$. \qed

\medskip

It is easy to see that step~\ref{step-1} (the bottleneck step) takes $O(mn)$ time per node. 
Hence the running time of the algorithm is $O(mn^2)$; thus Theorem~\ref{thm:main} follows.

\subsection{A simple extension of our algorithm: Algorithm {\sc MinMargin}}
\label{sec:relax-popularity}

Our algorithm can be extended to compute an arborescence with minimum {\em unpopularity margin} when nodes have weak rankings. 
When $D'$ does not admit an arborescence, algorithm \minmarg{} below computes a max-size branching $\tilde{B}$ in $D'$ and adds edges from the root $r$ 
to all root nodes in $\tilde{B}$ so as to make an arborescence of this branching in $D'$. This arborescence in $D'$ is then transformed into an arborescence in $D$ 
exactly as in our earlier algorithm.

\noindent
\begin{enumerate}
    \item Let $D'$ be the graph constructed in our algorithm for Theorem~\ref{thm:main}, and let $\tilde{B}$ be a branching of maximum cardinality in $D'$. 
    \item Let $B'=\{e \mid f(e) \in \tilde{B}\}$, $R_1=\{v \in V \mid \delta^-(v) \cap B' \neq \emptyset \}$, $R_2 = \emptyset$.
    \item For each $X \in \mathcal{X}'$ which is a root in the branching $\tilde{B}$, select one arbitrary $v \in V$ such that $X_v = X$, add $v$ to $R_2$ and $(r,v)$ to $B'$.
    \item For each $v \in R_1 \cup R_2$: let $A_v$ be an arborescence in $(X_v,S(X_v))$. 
    \item Return $A^*:= B' \bigcup_{v \in R_1 \cup R_2} A_v$. 
\end{enumerate}

\begin{reptheorem}{restate_thm:unpopmargin}[$\star$]
  \label{thm:unpopmargin}
  When nodes have weak rankings, Algorithm \minmarg{} returns an arborescence with minimum unpopularity margin in $D = (V \cup\{r\},E)$. 
\end{reptheorem}

\section{The Popular Arborescence Polytope of $D$}
\label{sec:polytope}
We now describe the popular arborescence polytope of $D = (V\cup\{r\},E)$ in $\mathbb{R}^m$. 
Throughout this section we assume that every node has a weak ranking over its incoming edges.
The arborescence polytope ${\cal A}$ of $D$ is described below~\cite{KoVy06a}.

\begin{eqnarray}
\sum_{e \in E[X]}x_e \ & \le & \ |X|-1 \ \ \ \ \ \ \ \ \ \forall\, X \subseteq V,\  |X| \ge 2. \label{constr1}\\
\sum_{e \in \delta^-(v)}x_e \ & = &\ 1 \ \ \ \forall\, v \in V \ \ \ \ \ \text{and} \ \ \ \ \ x_e \ \ge\ 0 \ \ \forall\, e \in E. \label{constr2}
\end{eqnarray}

We will define a subgraph $D^* = (V \cup \{r\}, E_{D^*})$ of $D$: this is essentially the {\em expanded} version of the graph $D'$ from our algorithm.
The edge set of $D^*$ is:

\smallskip
\hspace*{.2in}$E_{D^*} = \bigcup_{X \in {\cal X}'} S(X) \cup \{(u,v)\in E: X_v \in \mathcal{X}'$, $u \notin X_v$, and $(u,v)$ is \\ 
\hspace*{1.85in}undominated in $\delta^-(X_v)\}$.

\smallskip

Thus each  
{set} $X \in {\cal X}'$, {which is a node} in $D'$, is replaced in $D^*$ by 
the nodes in $X$ and with edges in $S(X)$ between nodes in $X$.
We also replace edges in $D'$  
between sets in ${\cal X}'$ by the original edges in $E$.

\begin{lemma}
\label{lem:pop-arb-structure}
If every node has a weak ranking over its incoming edges, then
every popular arborescence in $D$ is an arborescence in $D^*$ that includes exactly $|X|-1$ edges from $S(X)$ for each $X \in {\cal X}'$.
\end{lemma}
\begin{proof}
Let $A$ be a popular arborescence in $D$ and let $X \in \mathcal{X}'$. By Lemma~\ref{lem:poparb-enters-W-once} we know $|A \cap \delta^-(X)| = 1$; moreover, the proof of Theorem~\ref{thm:no-return-no-poparb} tells us that the unique edge in $A \cap \delta^-(X)$ is contained in $D^*$. 
So $A$ contains $|X|-1$ edges from $E[X]$ for each $X \in {\cal X}'$.  It remains to show that these $|X|-1$ edges are in $S(X)$.

Let $u \in X$. Suppose $A(u) \in E[X]\setminus S(X)$. This means that either (i)~$A(u)$ is dominated by some edge in $E[X] \cup \delta^-(X)$ or (ii)~$u$ is indifferent between $A(u)$ and some edge in $\delta^-(X)$. Let ${\cal Y}$ be a dual certificate of $A$.
We know that $Y_u \subseteq X_u \subseteq X$ (by Lemma~\ref{lem:witnessInWv}). Since the entry point of $A$ into $X$ is not in $Y_u$,
there is an edge $e \in S(X)\cap\delta^-(Y_u)$. 

Let $e$ enter $w \in Y_u$. Since $e \in S(X)$, we have $e \succ_w A(w)$ or $e \sim_w A(w)$, hence $c_A(e) \in \{0,1\}$. If $w \ne u$, then
$e$ enters two sets $Y_u$ and $\{w\}$---thus the constraint in \ref{LP2} corresponding to edge $e$ is violated. 
If $w = u$ then $e \succ_u A(u)$ (since $A(u) \in E[X]\setminus S(X)$, $e \in S(X)$, and $u$ has a weak ranking over its incoming edges): 
so $c_A(e) = 0$. Since $e$ enters one set $Y_u$, the constraint corresponding to $e$ in \ref{LP2} is again violated. So $A(u) \in S(X)$, i.e.,
$A \cap E[X] \subseteq S(X)$. \qed
\end{proof}

Hence, every popular arborescence in $D$ satisfies constraints~\eqref{constr1}-\eqref{constr2} along with constraints~\eqref{popular-constr}
given below, where $E_{D^*}$ is the edge set of $D^*$.
\begin{equation}
\label{popular-constr}
\sum_{e \in E[X]}x_e \ = \ |X|-1 \ \ \forall\, X \in {\cal X}',\ |X|\ge 2\ \ \ \ \text{and} \ \ \ \ \ x_e \ =\ 0 \ \ \forall\, e \in E \setminus E_{D^*}
\end{equation}

\vspace{-4pt}
\noindent
Note that constraints~\eqref{popular-constr} define a face ${\cal F}$ of the arborescence polytope ${\cal A}$ of $D$. Thus every popular arborescence 
in $D$ belongs to face ${\cal F}$.

Consider a vertex in face ${\cal F}$: this is an arborescence $A$ in $D$ of the form $A' \cup_{X\in{\cal X}'}A_X$ where
(i)~$A_X$ is an arborescence in $(X,S(X))$ whose root is an entry-point of $X$ and
(ii)~$A' = \{e_X: X \in {\cal X}'\}$ where $e_X$ is an edge in $D^*$ entering the root of $A_X$.
Theorem~\ref{thm:alg-rootedpop-safe} proved that such an arborescence $A$ is popular in $D$. 
Thus we can conclude Theorem~\ref{lem:polytope} which proves the upper bound in Theorem~\ref{thm:polytope}. 
The lower bound in Theorem~\ref{thm:polytope} is given in the \nameref{ch:appendix}. 
\begin{theorem}
\label{lem:polytope}
If every node has a weak ranking over its incoming edges, then
face ${\cal F}$ (defined by constraints~\eqref{constr1}-\eqref{popular-constr}) is the popular arborescence polytope of $D$.
\end{theorem}

A compact extended formulation of this polytope and all missing proofs are in the \nameref{ch:appendix}. 
We also discuss popular {\em mixed branchings} (probability distributions over branchings) there. 

\medskip

\noindent{\bf Acknowledgments.} Thanks to Markus Brill for helpful discussions on liquid democracy, and to Nika Salia for our  conversations in G\'ardony. Tam\'{a}s Kir\'{a}ly is supported by NKFIH grant no.~K120254 and by the HAS, grant no.~KEP-6/2017, Ildik\'{o} Schlotter was supported by NFKIH grants no. K 128611 and no. K 124171, and Ulrike Schmidt-Kraepelin by the Deutsche For\-schungs\-gemeinschaft (DFG) under grant BR~4744/2-1.

\bibliographystyle{splncs04}
\bibliography{abb,poparb_literature}

\begin{thebibliography}{10}
\providecommand{\url}[1]{\texttt{#1}}
\providecommand{\urlprefix}{URL }
\providecommand{\doi}[1]{https://doi.org/#1}

\bibitem{AIKM07}
Abraham, D.J., Irving, R.W., Kavitha, T., Mehlhorn, K.: Popular matchings. SIAM
  Journal on Computing  \textbf{37}(4),  1030--1034 (2007)

\bibitem{BIM10}
Bir{\'o}, P., Irving, R.W., Manlove, D.F.: Popular matchings in the marriage
  and roommates problems. In: Proceedings of the 7th International Conference
  on Algorithms and Complexity (CIAC), Lecture Notes in Computer Science
  (LNCS), vol.~6078, pp. 97--108. Springer (2010)

\bibitem{BGL18a}
Bloembergen, D., Grossi, D., Lackner, M.: On rational delegations in liquid
  democracy. In: Proceedings of the 33rd AAAI Conference on Artificial
  Intelligence (AAAI) (2019)

\bibitem{BlZu16a}
Blum, C., Zuber, C.I.: Liquid democracy: Potentials, problems, and
  perspectives. Journal of Political Philosophy  \textbf{24}(2),  162--182
  (2016)

\bibitem{Bril18a}
Brill, M.: Interactive democracy. In: Proceedings of the 17th International
  Conference on Autonomous Agents and Multiagent Systems (AAMAS) Blue Sky Ideas
  Track. pp. 1183--1187 (2018)

\bibitem{ChGr17a}
Christoff, Z., Grossi, D.: Binary voting with delegable proxy: {A}n analysis of
  liquid democracy. In: Proceedings of the 16th Conference on Theoretical
  Aspects of Rationality and Knowledge (TARK). pp. 134--150 (2017)

\bibitem{CCZ}
Conforti, M., Cornu{\'e}jols, G., Zambelli, G.: Integer Programming, Graduate
  Texts in Mathematics, vol.~271. Springer (2014)

\bibitem{CHK15}
Cseh, {\'A}., Huang, C.C., Kavitha, T.: Popular matchings with two-sided
  preferences and one-sided ties. SIAM Journal on Discrete Mathematics
  \textbf{31}(4),  2348 -- 2377 (2017)

\bibitem{CK16}
Cseh, {\'A}., Kavitha, T.: Popular edges and dominant matchings. Mathematical
  Programming  \textbf{172}(1),  209 -- 229 (2017)

\bibitem{Darm13b}
Darmann, A.: Popular spanning trees. International Journal of Foundations of
  Computer Science  \textbf{24}(5),  655 -- 677 (2013)

\bibitem{Darm16}
Darmann, A.: It is difficult to tell if there is a {Condorcet} spanning tree.
  Mathematical Methods of Operations Research  \textbf{84}(1),  94 -- 104
  (2016)

\bibitem{DKP11a}
Darmann, A., Klamler, C., Pferschy, U.: Finding socially best spanning trees.
  Theory and Decision  \textbf{70}(4),  511 -- 527 (2011)

\bibitem{Edmo67a}
Edmonds, J.: Optimum branchings. Journal of Research of the National Bureau of
  Standards  \textbf{71B}(4),  233 -- 240 (1967)

\bibitem{EGP19a}
Escoffier, B., Gilbert, H., Pass-Lanneau, A.: The convergence of iterative
  delegations in liquid democracy in a social network. In: Proceedings of the
  12th International Symposium on Algorithmic Game Theory (SAGT), Lecture Notes
  in Computer Science (LNCS), vol. 11801, pp. 284 -- 297. Springer (2019)

\bibitem{FK20}
Faenza, Y., Kavitha, T.: Quasi-popular matchings, optimality, and extended
  formulations. In: Proceedings of the 31th Annual ACM-SIAM Symposium on
  Discrete Algorithms (SODA) (2020), forthcoming

\bibitem{FKPZ19}
Faenza, Y., Kavitha, T., Powers, V., Zhang, X.: Popular matchings and limits to
  tractability. In: Proceedings of the 30th Annual ACM-SIAM Symposium on
  Discrete Algorithms (SODA). pp. 2790--2809 (2019)

\bibitem{Fulk74a}
Fulkerson, D.R.: Packing rooted directed cuts in a weighted directed graph.
  Mathematical Programming  \textbf{6}(1),  1 -- 13 (1974)

\bibitem{GS62}
Gale, D., Shapley, L.S.: College admissions and the stability of marriage. The
  American Mathematical Monthly  \textbf{69}(1),  9--15 (1962)

\bibitem{Gar75}
G{\"a}rdenfors, P.: Match making: {A}ssignments based on bilateral preferences.
  Behavioral Science  \textbf{20}(3),  166--173 (1975)

\bibitem{GKMP18}
G{\"o}lz, P., Kahng, A., Mackenzie, S., Procaccia, A.: The fluid mechanics of
  liquid democracy. In: Proceedings of the 14th International Workshop on
  Internet and Network Economics (WINE), Lecture Notes in Computer Science
  (LNCS), vol. 11316, pp. 188--202 (2018)

\bibitem{GMSZ19}
Gupta, S., Misra, P., Saurabh, S., Zehavi, M.: Popular matching in roommates
  setting is {NP}-hard. In: Proceedings of the 30th Annual ACM-SIAM Symposium
  on Discrete Algorithms (SODA). pp. 2810 -- 2822 (2019)

\bibitem{HL15}
Hardt, S., Lopes, L.: Google votes: A liquid democracy experiment on a
  corporate social network. Tech. rep., Technical Disclosure Commons (2015)

\bibitem{HK11}
Huang, C.C., Kavitha, T.: Popular matchings in the stable marriage problem.
  Information and Computation  \textbf{222},  180 -- 194 (2013)

\bibitem{HK17}
Huang, C.C., Kavitha, T.: Popularity, mixed matchings, and self-duality.
  Proceedings of the 28th Annual ACM-SIAM Symposium on Discrete Algorithms
  (SODA) pp. 2294--2310 (2017)

\bibitem{Kav14}
Kavitha, T.: A size-popularity tradeoff in the stable marriage problem. SIAM
  Journal on Computing  \textbf{43}(1),  52--71 (2014)

\bibitem{Kav16}
Kavitha, T.: Popular half-integral matchings. In: Proceedings of the 43rd
  International Colloquium on Automata, Languages, and Programming (ICALP),
  Leibniz International Proceedings in Informatics (LIPIcs), vol.~55, pp.
  22:1--22:13. Schloss Dagstuhl--Leibniz-Zentrum fuer Informatik (2016)

\bibitem{KMN09}
Kavitha, T., Mestre, J., Nasre, M.: Popular mixed matchings. Theoretical
  Computer Science  \textbf{412}(24),  2679--2690 (2011)

\bibitem{KoVy06a}
Korte, B., Vygen, J.: Combinatorial Optimization. Springer (2012)

\bibitem{KoRi18a}
Kotsialou, G., Riley, L.: Incentivising participation in liquid democracy with
  breadth first delegation. Tech. Rep. 1811.03710, arXiv (2018)

\bibitem{McC06}
{McCutchen}, R.M.: The least-unpopularity-factor and least-unpopularity-margin
  criteria for matching problems with one-sided preferences. In: Proceedings of
  the 8th Latin American Conference on Theoretical Informatics (LATIN), Lecture
  Notes in Computer Science (LNCS), vol.~4957, pp. 593--604. Springer (2008)

\end{thebibliography}
\newpage

\chapter*{Appendix}\label{ch:appendix}
\appendix
\section{Missing proofs from Sections~\ref{sec:algo} and \ref{sec:polytope}}
\repeatlemma{restate_lem:u-in-W_v}

\begin{proof}
We first show that $X_u^i \subseteq X_{v}^i$ for any $i$, where we set $X_v^{i}:=X_v$ whenever $X_v^{i}$ is not defined by the above algorithm.
The claim clearly holds for $i=0$. Let $i$ be the smallest index such that $x \in X_u^i \setminus X_{v}^i$ for some node $x$; 
we must have $x \in X_{u}^{i-1} \cap X_{v}^{i-1}$. By the definition of $X_{u}^{i}$, $x$ is reachable from $u$ in $S(X_{u}^{i-1})$. 
Note that $X_{u}^{i-1} \subseteq X_{v}^{i-1}$ implies $S(X_{u}^{i-1}) \subseteq S(X_{v}^{i-1})$, 
which yields that $x$ is reachable from $u$ in $S(X_{v}^{i-1})$ as well. Moreover, $u$ is reachable from $v$ in $S(X_{v}^{i-1}) \supseteq S(X_v)$ because $u \in  X_{v}$ and $S(\cdot)$ is monotone. Hence it follows that $x$ is reachable from $v$ in $S(X_{v}^{i-1})$ via $u$, contradicting the assumption that $x \notin X_{v}^{i}$. This proves the second statement of
the lemma. 

Now we will show the laminarity of $\mathcal{X}$. For contradiction, assume there exist $s,t \in V$ such that $X_s$ and $X_t$ {\em cross},
i.e., their intersection is non-empty, and neither contains the other.
Then, by the second statement of the lemma, neither $s \in X_{t}$ nor $t \in X_{s}$ can hold.
So we have that $s \notin X_t$ and $t \notin X_{s}$. 

Let $(x,y)$ be an edge in $S(X_{t})$ such that $y \in X_{s} \cap X_{t}$ but $x \in X_{t} \setminus X_{s}$; since each node in $X_{t}$ is reachable from $t$ in $S(X_{t})$, such an edge exists. Since $y \in X_{s} \setminus \{s\}$, there also exists an edge $(u,y)$ in $S(X_{s})$. As $x \notin X_{s}$ but $(u,y) \in S(X_{s})$, we know that $(u,y) \succ_y (x,y)$ which contradicts $(x,y) \in S(X_{t})$.~\qed
\end{proof}

\repeattheorem{restate_thm:alg-rootedpop-safe}
\begin{proof}
  We start by showing that $A^*$ is an arborescence in $D$. Then, we construct a dual certificate of value $n$ for $A^*$. This will prove the popularity of $A^*$.
  
  The laminarity of $\mathcal{X}$ implies that sets in ${\mathcal{X}}'$ are pairwise disjoint. Moreover, by construction, each node in $V$ is included in at least one set in $\mathcal{X}$, namely $v \in X_v$ for each $v \in V$. Hence, ${\mathcal{X}}'$ forms a partition of $V$. By construction of $A^*$, each node $v \in V$ can be reached in $A^*$ from the \emph{local root} $s \in R$ for which $v \in X_s$. It remains to show that the root $r$ reaches all local roots $s \in R$ in $A^*$. This can be shown by induction over the distance of $s$ from the root $r$ within the arborescence $A'$. It remains to show that $|A^*| =  n$. Let $k := |{\mathcal{X}}'|$. Since ${\mathcal{X}}'$ is a partition of $V$, we get that
  \begin{align*} \textstyle \vert A^* \vert = \vert A'\vert + \sum_{v \in R} \vert A_v \vert = k + \sum_{v \in R} ( \vert X_v \vert - 1) = k + n - k=n. \end{align*}

We turn to the second part of the proof and show a dual certificate of size $n$ for $A^*$. We claim that $\mathcal{Y}:= \{X_v: v \in R\} \cup \{\{v\}: v \in V \setminus R\}$ is such a dual solution. Note that $|\mathcal{Y}| = |R| + |V\setminus R| = n$. We now show that for all $v \in V$, the incoming edges satisfy the constraints in \ref{LP2}. 

Suppose $v \in R$. An edge $(w,v) \in E$ enters one set of $\mathcal{Y}$ iff $w \not\in X_v$ and no set iff $w \in X_v$. Hence, it suffices to show that $c_{A^*}((w,v)) \in \{1,2\}$ for $w \notin X_v$. Let $(u,v)$ be the incoming edge of $v$ in arborescence $A^*$; note that $(u,v) \in A'$ and $u \notin X_v$. By construction of ${E}'$, $(w,v)$ does not dominate $(u,v)$ and therefore $c_{A^*}((w,v)) \in \{1,2\}$. The same argument
works for $v \in V\setminus R$ and $\{v\} \in {\cal X}'$.

Suppose $v \in V \setminus R$. Let $s$ be $v$'s local root, i.e., the unique $s \in R$ with $v \in X_s$. Then $(u,v) \in A_s \subseteq S(X_s)$ by construction of $A_s$.
Any edge $(w,v) \in \delta^-(v)$ enters at most two sets of $\mathcal{Y}$:
$\{v\}$ and possibly $X_s$.
If, on the one hand, $(w, v) \in \delta^-(X_s)$, then $(u,v) \in S(X_s)$ dominates $(w,v)$ by property~\ref{property-2} of $S(X_s)$, and hence $c_{A^*}((w,v)) = 2$. 
If, on the other hand, $w \in X_s$, then 
$(u,v) \in S(X_s)$ is not dominated by $(w,v)$ by property~\ref{property-1} of $S(X_s)$,
and hence $c_{A^*}((w,v)) \geq 1$.
This completes the proof that $\mathcal{Y}$ is a dual certificate of size~$n$ for $A^*$, thus $A^*$ is popular. \qed
\end{proof}

\repeatlemma{restate_lem:witnessInWv}
\begin{proof}
If $Y_v=\{v\}$, then $Y_v \subseteq X_v$ is trivial, so suppose that $Y_v$ is not a singleton.
We know from Corollary~\ref{cor:singletons} that $Y_w$ is a singleton set for each $w \in Y_v \setminus \{v\}$. Moreover, for every $(u,w) \in A$ with $w \in Y_v \setminus \{v\}$ it holds that $u \in Y_v$ since this edge would otherwise enter two sets; however, $c_A((u,w))=1$ as $(u,w) \in A$. 

Assume for contradiction that $Y_v \setminus X_v \neq \emptyset$. 
Let $i$ be the last iteration when $Y_v \subseteq X_v^i$. Then there exists a subset of $Y_v$ which is not reachable by edges in $S(X_v^i)$, i.e., $\delta^{-}(Y_v \setminus X_v^{i+1}) \cap S(X_v^{i}) = \emptyset$. On the other hand, we know that the arborescence $A$ can only enter nodes in $Y_v\setminus \{v\}$ by edges from $E[Y_v]$, and therefore, it needs to contain at least one edge from $\delta^{-}(Y_v \setminus X_v^{i+1}) \cap \delta^{+}(X_v^{(i+1)})$. Let $(u,w)$ be this edge (see Fig.~\ref{fig:lemma5}). By construction of $X_v^i$ and $X_v^{i+1}$, we know that one of the following cases has to be true. 

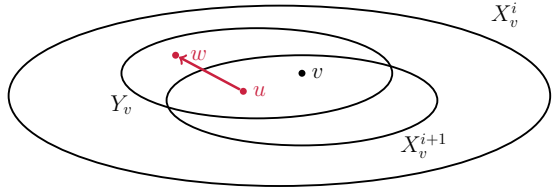
\begin{figure}[h]
\centering
\scalebox{.6}{
\begin{tikzpicture}
\draw[very thick] (.5,-.5) ellipse (6cm and 2cm);
\draw[very thick] (0,0) ellipse (3cm and 1cm);
\draw[very thick] (1,-.6) ellipse (3cm and 1cm);

\node(wi) at (5.5,1.3){\Large $X_v^{i}$};
\node(wi1) at (3.7,-1.6){\Large $X_v^{i+1}$};
\node(yv) at (-3,-.7){\Large $Y_v$};

\node[circle, black, fill, inner sep=.4ex](v) at (1,0){}; \node[right= 3pt] at (v){\Large $v$};

\node[circle, inner sep = .4ex, red,fill](u) at (-.3,-.4){}; \node[right= 3pt,red] at (u){\Large $u$};
\node[circle, inner sep = .4ex, red,fill](w) at (-1.8,.4){}; \node[right= 6pt,red] at (w){\Large $w$};

\draw[->, red, ultra thick] (u) -- (w);
\end{tikzpicture}
}
\caption{Illustration of the situation in the proof of Lemma \ref{lem:witnessInWv}.}\label{fig:lemma5}
\end{figure}

\textbf{Case 1}. There exists an edge $(x,w) \in E[X_v^{i}]$ which dominates $(u,w)$. Note that we do not know if $(x,w) \in E[Y_v]$ or not. However, $c_A((x,w)) = 0$ in either case, but by Corollary~\ref{cor:singletons}, $(x,w)$ enters at least one set in $\mathcal{Y}$, namely $\{w\}$. This is a violation of \ref{LP2} and it contradicts $\mathcal{Y}$ being a dual certificate for $A$.

\smallskip

\textbf{Case 2}. There exists an edge $(x,w) \in \delta^{-}(X_v^{i})$ which is not dominated by $(u,w)$. Note that $c_A((x,w)) \in \{0,1\}$, but $(x,w) \in \delta^{-}(Y_v)$ and so the edge~$(x,w)$ enters two dual sets: $Y_v$ and $\{w\}$. This contradicts $\mathcal{Y}$ being a dual solution. \qed
\end{proof}

\repeatlemma{restate_lem:poparb-enters-W-once}
\begin{proof}
Let $X \in \cal{X}'$ and let $A$ be a popular arborescence which enters $X$ at some node $v \in V$ through an edge $(u,v) \in A \cap \delta^-(X)$. 
Moreover, let $\mathcal{Y}$ be a dual certificate for $A$, and let $Y_v$ be the set whose entry-point is $v$.

Let $\entry(X) := \{w \in V: X_w = X\}$. We first show that $\entry(X) \subseteq Y_v$. Assume for contradiction that there exists 
$w \in \entry(X)$ such that $w \notin Y_v$. Since $X_w = X$ we know that there exists a $w$-$v$ path $P$ in $(X,S(X))$. Hence, there exists
an edge $e \in P$ which enters $Y_v$. If the head of $e$ is $v$, we know that $e$ dominates $(u,v) \in \delta^-(X)$ and hence $c_A(e)=0$, a contradiction 
to the feasibility of $\mathcal{Y}$. If $v$ is not the head of $e$, then $e$ not only enters $Y_v$, but also the singleton set 
corresponding to its head. However, $c_A(e) \leq 1$ since $e$ is an undominated edge by $e \in S(X)$, a contradiction to the feasibility of $\mathcal{Y}$. 

To prove that $v \in \entry(X)$, let us choose some $s \in \entry(X)$. By the previous paragraph and Lemma~\ref{lem:witnessInWv}, we get 
$s \in Y_v \subseteq X_v$, from which Lemma~\ref{lem:u-in-W_v} implies $X_s \subseteq X_v$. Because $s \in \entry(X)$, we have $X = X_s \subseteq X_v$. Because $X \in \mathcal{X}'$ is inclusionwise maximal in $\mathcal{X}$, we get $X = X_v$, proving $v \in \entry(X)$.

It remains to prove that $A$ enters $X$ only once. 
Suppose for contradiction that there exist two nodes $v,v' \in \entry(X)$ such that $(u,v), (u',v') \in A \cap \delta^-(X)$. 
By $\emptyset \neq \entry(X) \subseteq Y_v \cap Y_{v'}$ and the laminarity of $\mathcal{Y}$, 
we can assume w.l.o.g. that $Y_v \subseteq Y_{v'}$. Moreover, since $Y_{v'} \subseteq X$, the arborescence edge $(u,v)$ enters both
$Y_v$ and $Y_{v'}$, a contradiction to the feasibility of the dual solution $\mathcal{Y}$. \qed
\end{proof}

\paragraph{\bf  Lower bound for the popular arborescence polytope of $D$.}
Let $D = (V \cup\{r\},E)$ be the complete graph where every node $v \in V$ regards all other nodes $u \in V$ as top-choice in-neighbors and
$r$ as its second-choice in-neighbor.
Here ${\cal X}' = \{V\}$ and 
$D^*$ is the complete bidirected graph on $V$ along with edges $(r,v)$ for all $v \in V$. 
We claim that
in any minimal system contained in \eqref{constr1}-\eqref{popular-constr}, the constraint $\sum_{e \in E[X]}x_e  \le |X|-1$ for
every $X \subset V$ with $|X| \ge 2$ has to be present. This is because a cycle on the nodes in $X$ along with any {rooted}
arborescence $A$ on 
$V \setminus X$ plus $(r,v)$, where $v$ is the root of $A$, satisfies all the remaining constraints. Thus any minimal 
system of inequalities from \eqref{constr1}-\eqref{popular-constr} has to contain $2^n - n - 2$ inequalities from \eqref{constr1}:
one for every $X \subset V$ with $|X| \ge 2$.
Since inequalities in a minimal system are in one-to-one correspondence with the facets of the polyhedron they 
describe~\cite[Theorem~3.30]{CCZ}, the lower bound given in Theorem~\ref{thm:polytope} follows.

\paragraph{\bf A compact extended formulation.}
We now describe a compact extended formulation of the popular arborescence polytope of $D$ when node preferences are weak rankings. 
We know from Lemma~\ref{lem:pop-arb-structure} that every popular arborescence in $D$ is an arborescence in $D^*$ that includes exactly $|X|-1$ edges 
from $S(X)$ for each $X \in {\cal X}'$. Conversely, any such arborescence in $D^*$ is a popular arborescence in $D$ 
(by Theorem~\ref{thm:alg-rootedpop-safe}).

Thus the popular arborescence polytope of $D$ is the face of the arborescence polytope of $D^*$ that corresponds to the constraints
$\sum_{e \in E_{D^*}[X]}x_e  = |X|-1$ for all $X \in {\cal X}'$. 
Let ${\cal A}_{D^*}$ be the arborescence polytope of $D^* = (V \cup \{r\}, E_{D^*})$.
We will now use a compact extended formulation of ${\cal A}_{D^*}$.

Recall that $|V| = n$. Let ${\cal P}_{D^*}$ be the polytope defined by constraints~\eqref{new-constr1}-\eqref{new-constr4} on
variables $x_e, f_e^v$ for $e \in E_{D^*}$ and $v \in V$. 
It is known~\cite{CCZ} that ${\cal P}_{D^*}$ is a compact extended formulation of the arborescence polytope ${\cal A}_{D^*}$.  
Note that ${\cal A}_{D^*}$ is the projection of ${\cal P}_{D^*}$ on to $x$-space. 
\begin{eqnarray}
x_e \ & \ge & \ f^v_e \ \ \ge \ \ 0 \ \ \ \forall v \in V \ \text{and}\ e \in E_{D^*}\label{new-constr1}\\
\sum_{e\in\delta^+(r)}f^v_e \ & = & \ 1 \ \ \ \ \ \ \ \ \ \ \ \ \ \ \forall v \in V \label{new-constr2}\\
\sum_{e\in\delta^+(u)}f^v_e \ - \ \sum_{e\in\delta^-(u)}f^v_e \ & = & \ 0 \ \ \ \ \ \ \ \ \ \ \ \ \ \ \forall u, v \in V, \ u \ne v \label{new-constr3}\\
\sum_{e\in E_{D^*}}x_e \ & = & \ n. \label{new-constr4}
\end{eqnarray}

For any $X \subseteq V$ with $|X| \ge 2$, the constraint $\sum_{e\in E_{D^*}[X]}x_e \le |X|-1$ is a valid inequality 
for ${\cal A}_{D^*}$ and also for ${\cal P}_{D^*}$. Thus the intersection of ${\cal A}_{D^*}$ along with the tight constraints 
$\sum_{e\in E_{D^*}[X]}x_e = |X|-1$ for all $X \in {\cal X}'$ is a face of ${\cal A}_{D^*}$. Call this face ${\cal F}_{D^*}$---this is the
popular arborescence polytope of $D$. 

\smallskip

Consider the face of ${\cal P}_{D^*}$ that is its intersection with $\sum_{e\in E_{D^*}[X]}x_e = |X|-1$ for all $X \in {\cal X}'$. 
This face of ${\cal P}_{D^*}$ is an extension ${\cal F}_{D^*}$. The total number of constraints used to describe this face of ${\cal P}_{D^*}$ is $O(mn)$.

\section{Branchings with minimum unpopularity margin}
 
Recall the definition of the {\em unpopularity margin} for branchings from Section~\ref{sec:intro}. Again,
instead of studying minimum unpopularity margin branchings within the digraph $G$, we look at $r$-arborescences of minimum unpopularity margin within the digraph $D$. 
It is easy to see that the unpopularity margin of a branching in $G$ is the same as the unpopularity margin of the corresponding arborescence in $D$.\footnote{Note that, due to the special structure of $D$, there always exists an arborescence $A'$ such that $A' \in \argmax_{B \in \mathcal{B(D)}}\phi(B,A)-\phi(A,B)$, where $\mathcal{B}(D)$ is the set of branchings in $D$.} Thus we are looking for an arborescence of minimum unpopularity margin in $D$.

Furthermore, recall that by Proposition~\ref{prop0} and Lemma~\ref{lem:dual-characterization} in Section~\ref{sec:lp}, the unpopularity margin $\mu(A)$ of an arborescence $A$ fulfills
\[\mu(A) = n - c_A(A') = n - |\mathcal{Y}|,\]
 where $A'$ is a min-cost arborescence in $D$ with respect to $c_A$ and $\mathcal{Y}$ is a dual certificate of maximum cardinality for $A$.

\repeattheorem{restate_thm:unpopmargin}
Algorithm \minmarg{} is described in Section \ref {sec:relax-popularity}. To prove Theorem~\ref{thm:unpopmargin}, we first show in Lemma~\ref{lem:l-roots-means-l-1-margin} that the size of the maximum cardinality branching $\tilde{B}$ 
bounds the unpopularity margin of the arborescence $A^*$ returned by Algorithm \minmarg{}. 
Then we provide Lemma~\ref{lem:witnessInWv_weakorder} and Observation~\ref{obs:source-of-subtrees} that will be helpful in showing the optimality of Algorithm \minmarg{}.

\begin{lemma}
\label{lem:l-roots-means-l-1-margin}
If the number of roots in the branching $\tilde{B}$ is $\ell$, then aborescence $A^*$ returned by Algorithm \minmarg{} has unpopularity margin at most $\ell-1$. 
\end{lemma}

\begin{proof}
We show that $A^*$ has unpopularity margin at most $\ell-1$  by constructing a dual certificate of size $n-\ell+1$; by Lemma~\ref{lem:dual-characterization} this is sufficient. Define $\mathcal{Y}:= \{X_v \mid v \in R_1\} \cup \{\{v\} \mid v \in V \setminus \{R_1 \cup R_2\}\}$. It is easy to see that $\mathcal{Y}$ contains $n-(|R_2|-1)$ = $n-\ell+1$ sets ($r \in R_2$ but $r\notin V$). It remains to show that any edge $(w,v)$ satisfies the constraints in \ref{LP2}; the argumentation is analogous to the one in the proof for Theorem~\ref{thm:alg-rootedpop-safe}.

First, if $v \in R_2$, then $v$ is not contained in any set of $\cert$, so no constraints are violated by $(w,v)$.
Otherwise, let $(u,v)$ be an incoming edge of $v$ in $A^*$.

Suppose $v \in R_1$; then $(u,v) \in B'$ and $u \notin X_v$. 
Edge $(w,v)$ enters one set of $\mathcal{Y}$ iff $w \not\in X_v$ and no set iff $w \in X_v$. 
Hence, it suffices to show that $c_{A^*}((w,v)) \in \{1,2\}$ for $w \notin X_v$. 
By construction of ${E}'$ (recall that $E'$ is the edge set of $D'$), $(w,v)$ does not dominate $(u,v)$ and therefore $c_{A^*}((w,v)) \in \{1,2\}$.

Suppose now $v \in V \setminus (R_1 \cup R_2)$. Let $s$ be $v$'s local root, i.e., $s \in R_1 \cup R_2, v \in X_s$; then $(u,v) \in A_s$.
Edge $(w,v)$ enters two sets of $\mathcal{Y}$ iff $w \not\in X_s$ and one set iff $w \in X_s$. 
If, on the one hand, $w \notin X_s$, then by construction of $A_s$ and property~\ref{property-2} of $S(X_s)$, it holds that $(w,v)$ is dominated by $(u,v)$, 
and hence $c_{A^*}((w,v)) = 2$. 
If, on the other hand, $w \in X_s$, then by construction of $A_s$ and property~\ref{property-1} of $S(X_s)$, 
$(w,v)$ does not dominate $(u,v)$, and hence $c_{A^*}((w,v)) \in \{1,2\}$. 
Thus, any edge satisfies the constraints in \ref{LP2} and  $\cert$ is a dual certificate for $A^*$. \qed
\end{proof}

Let $S \subseteq V$, $s \in S$ and $A \subseteq E$ be an arborescence rooted at $s$ and spanning exactly the nodes in $S$. 
We say that $A$ is \emph{locally popular with respect to} $S$, if the set family $\cert := \{\{v\}\mid v \in S \setminus \{s\}\} \cup \{S\}$ 
fulfills the constraints of the dual LP induced by $c_A$, where we set $c_A(e):=1$ for each edge $e \in \delta^-(s)$  and $c_A(e):=0$ for every $e \in \delta^-(v)$ with $v \in V\setminus S$ (see \ref{LP2}).

\begin{lemma}
\label{lem:witnessInWv_weakorder}
Let $S \subseteq V$ such that there exists an arborescence $A \subseteq E$ which is rooted at $v\in S$ and locally popular with respect to $S$. 
Then, $S \subseteq X_v$. 
\end{lemma}

\begin{proof}
The proof of this lemma is a direct analog of the proof of Lemma~\ref{lem:witnessInWv}: 
substituting $S$ for $Y_v$ and using the definition of local popularity instead of popularity, 
one can use the same arguments to obtain the statement of this lemma. \qed
\end{proof}

\begin{observation}
\label{obs:source-of-subtrees}
Let $\tilde{B}$ be a branching of maximum cardinality in $D'$ and $T \subseteq \tilde{B}$ be a maximal subarborescence of $\tilde{B}$ not containing $r$. 
Then, there exists $S \subseteq V(T)$ such that $\delta^{-}_{D'}(S) = \emptyset$.
\end{observation} 

\begin{proof}
Assume for contradiction that $\delta^{-}_{D'}(S) \neq \emptyset$ for all $S \subseteq V(T)$. 
Hence, every $X \in V(T)$ is reachable from $\{r\} \cup \mathcal{X} \setminus V(T)$ in $D'$. 
Consequently, we can modify $\tilde{B}$ by attaching each $X \in V(T)$ to some node in $\{r\} \cup \mathcal{X}' \setminus V(T)$, one by one. 
This contradicts the maximality of $\tilde{B}$.
\qed
\end{proof}

 Let $A$ be any arborescence, and $\cert$ a \emph{dual certificate} for $A$.
 Since $c_A(e)=1$ for every $e \in A$, we know that each edge in $A$ enters at most one set in $\mathcal{Y}$. 
 If an edge $(u,v) \in A$ enters a set of $\cert$, we refer to this set as $Y_v$, and we say that $Y_v$ \emph{belongs} to $v$ in $\cert$. 
 In contrast to the case of popular arborescences, it can be the case that the same set belongs to two edges in $\cert$, i.e., $Y_v=Y_{v'}$ but $v \neq v'$. 
 We say that $\cert$ is \emph{complete} on $S \subseteq V$, if $|\{Y_v \mid v \in S \}| = |S|$; 
 this concept will be crucial in the proof of Theorem~\ref{thm:unpopmargin}.
 By a simple counting argument we obtain that if $\cert$ is complete on $S$, then $v,v' \in S, v\neq v'$ implies $Y_v \neq Y_{v'}$.

\begin{proof}[of Theorem~\ref{thm:unpopmargin}]
By Lemma~\ref{lem:l-roots-means-l-1-margin}, the algorithm returns an arborescence with unpopularity margin at most $\ell-1$, 
where $\ell$ is the number of maximal subtrees in $\tilde{B}$. 
Let $A$ be an arborescence with minimum unpopularity margin and $\cert$ a corresponding dual certificate. 

Take any maximal subtree $T$  of $\tilde{B}$ not containing $r$. 
By Observation~\ref{obs:source-of-subtrees}, there exists some $S \subseteq V(T)$ with $\delta^{-}_{D'}(S)=\emptyset$. 
Below we prove that $\cert$ is not complete on $S^* := \bigcup_{X \in S} X$. 
As there are $\ell-1$ maximal subtrees of $\tilde{B}$ not containing $r$, and each contains a set of nodes on which $\cert$ is not complete, 
we get $|\cert| \leq n-(\ell-1)$. This implies $\mu(A) \geq \ell-1$ by Lemma~\ref{lem:dual-characterization}, proving the theorem.

It remains to show that $\cert$ is not complete on $S^*$. 
Assume for contradiction that $\cert$ contains a set $Y_x$ belonging to each $x \in S^*$.
Recall that $D^*$ is the expanded version of $D'$.
Note that, by the construction of $\mathcal{X}$ in the algorithm, a most preferred edge $(u,v)$ can enter a set $X\in \mathcal{X}$ only at a candidate entry node. Thus,
$\delta^{-}_{D^*}(S^*)=\emptyset$ means that $S^*$ is not entered by any most preferred edge.

Since $A$ enters $S^*$ but $\delta^{-}_{D^*}(S^*)=\emptyset$, there exists $(u,v) \in A \cap \delta^{-}_{D}(S^*)$ which is not included in $D^*$. 

\begin{claim}
$Y_v \cap S^* \not\subseteq X_v$
\end{claim}
\begin{proof}
Let $X \in S$ be the contracted node  entered by $(u,v)$. 

\smallskip

\textbf{Case 1}: $v \notin \entry(X)$. 
Let $s \in \entry(X)$. Then there exists an $s$-$v$-path $P$ in $(X,S(X))$; recall that every edge on $P$ is most preferred and dominates all edges entering $X$. If $s \notin Y_v$, then there is an edge $(u', v') \in P$ entering $Y_v$. If $v' \neq v$, then the most-preferred edge $(u', v')$ crosses two sets of $\cert$, a contradiction. If $v' = v$, then $(u', v')$ dominates $(u, v)$ but crosses $Y_v$, again a contradiction. We conclude that $s \in Y_v$. However, as $v$ is not in $\entry(X)$, by Lemma~\ref{lem:u-in-W_v} we know that $s \notin X_v$. We obtain $Y_v \cap S^* \not\subseteq X_v$.   

\smallskip

\textbf{Case 2}:  $v \in \entry(X)$, i.e., $X = X_v$. Since $(u,v) \notin D^*$, there exists an edge $(u',v) \in D^*$ which dominates $(u,v)$ and $u' \in V \setminus X$. Hence, $(u',v)$ must not enter any set in $\cert$ and we obtain $u' \in Y_v$. Clearly, we get that $Y_v  \cap S^* \not\subseteq X_v$. 
\qed
\end{proof}

Now, we are going to show that $A$ induces a locally popular arborescence on $Y_v \cap S^*$, rooted at $v$. 
By the above claim, this contradicts Lemma \ref{lem:witnessInWv_weakorder}. 

Consider a node $x \in Y_v \cap S^*\setminus\{v\}$. If $f=(w,x)$ is a most-preferred edge in $x$, then $w \in Y_v \cap S^*$: 
indeed, $f$ cannot enter $S^*$ because $\delta^-_{D^*}(S^*)=\emptyset$, and moreover, $f$ cannot enter $Y_v$ because $c_A(f) \leq 1$ and 
thus cannot enter both $Y_v$ and $Y_x$ (recall that $\cert$ is complete on $S^*$ which contains $x$, so $\cert$ contains a set $Y_x$ corresponding to $x$).
If $x$ prefers $f$ to $A(x)$, then $c_A(f)=0$ and thus $f$ cannot enter $Y_x$, implying $w \in Y_x$. 
However, this contradicts the fact that $\cert$ is two-layered: since $w \in S^*$, there exists a set $Y_w \in \cert$ corresponding to $w$, 
and so $w$ is contained in three sets of $\cert$, $Y_w$, $Y_x$ and $Y_v$.
Hence, we obtain that $A(x)$ must be a most-preferred edge for $x$. Since this holds for each $x \in Y_v \cap S^*\setminus\{v\}$, 
$A':=A \cap E[Y_v \cap S^*]$ is an arborescence rooted at $v$, containing only most-preferred edges. 

It remains to show that
$\cert':= \{Y_v \cap S^*\} \cup \{\{u\} \mid u \in (Y_v \cap S^*) \setminus \{v\}\}$ fulfills all constraints in \ref{LP2} w.r.t. $A'$ (with $c_{A'}(e)=1$ for each $e \in \delta^-(v)$). Observe that we need to verify this only for edges that point from $Y_v \setminus S^*$ to $Y_v \cap S^*$, as all other edges enter the same number of sets in $\cert'$ as in $\cert$. So let $f=(w,x)$ be such an edge. If $x=v$, then $f$ enters only $Y_v \cap S^*$ from $\cert'$; by $c_{A'}(f)=1$ this satisfies \ref{LP2}. If $x \neq v$, then $f$ enters two sets $Y_v \cap S^*$ and $\{x\}$ from $\cert'$. 
Since $\delta^-_{D^*}(S^*)=\emptyset$, we know that $f$ is not a most-preferred edge, so $x$ prefers $A'(x)$ to $f$, yielding $c_{A'}(f)=2$; note that here we need that $\succ_x$ is a weak ordering. This proves that all edges satisfy the constraints in \ref{LP2}, so we can conclude that $A'$ is indeed locally popular and spans $Y_v \cap S^*$. \qed
\end{proof}

The following theorem shows that Algorithm \minmarg{} cannot be extended for the case where each node $v$ has a partial order over $\delta^-(v)$.

\begin{theorem}
Given a directed graph where each node has a partial preference order over its incoming edges and an integer $k \leq n$, it is $\mathsf{NP}$-hard to decide if there exists a branching with unpopularity margin at most~$k$. \label{thm:hardness_unpopmargin}
\end{theorem}

\begin{proof}
We reduce from \textsc{$3$D-Matching} where we are given disjoint sets $X,Y,Z$ of equal cardinality and $T \subseteq X \times Y \times Z$, and we ask whether there exists $M \subseteq T$ with $|M| = |X|$ such that for distinct $(x,y,z), (x',y',z') \in M$ it holds that $x\neq x', y \neq y'$ and $z\neq z'$; such an $M$ is called a $3$D-\emph{matching}. W.l.o.g. we assume that $|X| > 3$ and every $x \in X\cup Y \cup Z$ is in some $t \in T$. 

We construct a digraph $D=(V\cup\{r\},E)$ together with a partial order $\succ_v$ over the incoming edges of $v$ for each $v \in V$ as follows. For every $x \in X\cup Y \cup Z$ we introduce a \emph{node gadget} consisting of a lower node $x_l$ and an upper node $x_u$. There exist two parallel edges, $d_x^{(1)}$ and $d_x^{(2)}$, from $x_u$ to $x_l$, 
and there exist two parallel edges, $r_x^{(1)}$ and $r_x^{(2)}$, from $r$ to $x_l$. 
Moreover, the upper node $x_u$ has an incoming edge from the upper node of every other node gagdet, i.e., $(x'_u,x_u) \in E$ for all $x' \in X\cup Y \cup Z \setminus \{x\}$. Lastly, there exists an incoming edge from $r$ to the upper node which we call $r_x^{(3)}$. 

For each $t \in T$ we introduce a \emph{hyperedge gadget} consisting of six edges in~$D$. More precisely, for each $x \in t$ we introduce two parallel edges from $x_l$ to $x_u$ which we call $t_x^{(1)}$ and $t_x^{(2)}$. This finishes the definition of $D$.

Let us now define the preferences $\{ \succ_v \mid v \in V\}$. 
A lower node $x_l$ has the following preferences over its incoming edges: 
\[d_x^{(1)}\succ r_x^{(1)} \text{, }\qquad d_x^{(2)} \succ r_x^{(2)},\] and all other pairs are not comparable. 
Let $t=(x,y,z)\in T $ and $\bar{t}:=\{x,y,z\}$. The preferences of an upper node $x_u$ are as follows: 
$$\begin{array}{lcl}
 (x'_u,x_u) \succ r_x^{(3)}      &\phantom{a}&\text{for each } x' \in X\cup Y \cup Z \setminus \{x\}, \\
    t_x^{(1)} \succ (x'_u, x_u)  &&\text{for each } x' \in X \cup Y \cup Z \setminus \bar{t},\\
    t_x^{(2)} \succ (x'_u, x_u)  &&\text{for each } x' \in \bar{t}\setminus\{x\},\\
   t_x^{(1)} \succ r_x^{(3)},    && t_x^{(2)} \succ r_x^{(3)},
\end{array} $$
and all other pairs are not comparable. See Figure~\ref{fig:3d-matching} for an illustration.

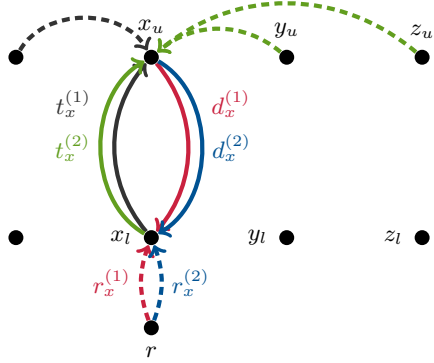
\begin{figure}[h]
\centering
\begin{tikzpicture}[scale=0.6,vertex/.style={circle,fill,inner sep=2pt}]
\node[vertex](r) at (3,-2){};\node[below=4pt] at (r){$r$};
\node[vertex](wu) at (0,4){};
\node[vertex](wl) at (0,0){};
\node[vertex](xu) at (3,4){};\node[above=6pt] at (xu){$x_u$};
\node[vertex](xl) at (3,0){};\node[left=4pt] at (xl){$x_l$};
\node[vertex](yu) at (6,4){};\node[above=4pt] at (yu){$y_u$};
\node[vertex](yl) at (6,0){};\node[left=4pt] at (yl){$y_l$};
\node[vertex](zu) at (9,4){};\node[above=4pt] at (zu){$z_u$};
\node[vertex](zl) at (9,0){};\node[left=4pt] at (zl){$z_l$};

\draw[red,densely dashed, bend left=20, ultra thick,->] (r) to node{} node[left] {$r_x^{(1)}$} (xl);
\draw[blue, densely dashed, bend right=20, ultra thick,->] (r) to node{} node[right] {$r_x^{(2)}$} (xl);

\draw[red,bend left=40, ultra thick,->, near start] (xu) to node{} node[right=8pt] {$d_x^{(1)}$} (xl);
\draw[blue, bend left=60, ultra thick,->] (xu) to node{} node[right] {$d_x^{(2)}$} (xl);

\draw[black!80,bend left=40, ultra thick,->, near end] (xl) to node{} node[left=8pt] {$t_x^{(1)}$} (xu);
\draw[green, bend left=60, ultra thick,->] (xl) to node{} node[left] {$t_x^{(2)}$} (xu);

\draw[green,densely dashed, bend right=40, ultra thick,->] (yu) to node{} node[left] {} (xu);
\draw[green, densely dashed, bend right=40, ultra thick,->] (zu) to node{} node[left] {} (xu);
\draw[black!80, densely dashed, bend left=60, ultra thick,->] (wu) to node{} node[left] {} (xu);

\end{tikzpicture}
\caption{Construction within the reduction for Theorem \ref{thm:hardness_unpopmargin}. A solid edge of a certain color dominates the dashed edge(s) of the same color; the figure assumes $(x,y,z) \in T$.}
\label{fig:3d-matching}
\end{figure}
Note that the digraph $D$ has the special property that every node $v \in V$ has at least one incomming edge from $r$. As a consequence of this structure, any branching $B$ in $D$ minimizing $\mu(B)$ must in fact be an arborescence rooted at $r$. Moreover, we can apply Lemma~\ref{lem:dual-characterization} to any given arborescence $A$ in $D$ as usual. 
In the following we show that there exists a $3$D-matching $M\subseteq T$ with $|M|=|X|$ iff there exists an $r$-arborescence in $D$ with unpopularity margin at most $2|X|$.

First, let $M \subseteq T$ be a $3$D-matching with $|M|=|X|$. We construct an arborescence $A$ together with a feasible dual certificate $\cert$ with $|\cert|=4|X|$. By Lemma~\ref{lem:dual-characterization}, this suffices to show that $A$ has unpopularity margin at most $6|X|-4|X|=2|X|$. We define \[A:= \{r_x^{(1)} \mid \text{for all } x \in X\cup Y \cup Z \} \cup \{t_w^{(1)} \mid \text{for all } t \in M \text{ and for all } w \in t\}\] and \[\mathcal{Y}:=\{\{x_u\} \mid \text{for all } x \in X \cup Y \cup Z \} \cup \{\{x_u,y_u,z_u,x_l,y_l,z_l\}\mid \text{for all } (x,y,z) \in M\}.\] 
Clearly, $A$ is indeed a $r$-arborescence. It remains to show that $\mathcal{Y}$ is a feasible dual solution. First consider a node $x_l$ for $x \in X \cup Y \cup Z$ which has four incoming edges. The edges $d_x^{(1)}$ and $d_x^{(2)}$ do not enter any set in $\cert$ and hence do not violate any constraint in the dual LP. Moreover, since node $x_l$ is indifferent between $r_x^{(1)}$ and $r_x^{(2)}$, we obtain $c_A(r_x^{(1)}) = c_A(r_x^{(2)}) = 1$ and hence, none of the corresponding constraints is violated. 

Now, consider $x_u$ for $x \in X \cup Y \cup Z$. Let $(x,y,z)$ be the hyperedge in $M$ containing $x$. 
We obtain that $c_A((y_u,x_u))=c_A((z_u,x_u))=c_A(t_x^{(2)})=1$ 
while for any other incoming edge $e$ of $x_u$ we get $c_A(e)=2$. By construction of $\cert$, none of the constraints is violated. This suffices to show the first direction of the equivalence. 

Now, let $A$ be an $r$-arborescence of unpopularity margin at most $2|X|$. Let $\cert$ be a corresponding laminar certificate of size $|\cert| = 4|X|$. 

Our first observation is that $x_l$ for any $x \in X\cup Y\cup Z$ is included in at most one set in $\cert$. This can be seen by a simple case distinction over the incoming edge of $x_l$ in $A$. No matter which of the four incoming edges to $x_l$ is selected in $A$, it always holds that $c_A(r_x^{(1)})=1$ or $c_A(r_x^{(2)})=1$, while both of them enter two sets of $\cert$. 

The first observation directly implies that a node gadget can intersect with at most two sets from $\cert$. Since the number of sets is greater than the number of node gadgets, there exist node gadgets which intersect with more than one set from $\cert$. Let $x \in X\cup Y \cup Z$ be a node such that the corresponding node gadget intersects with two sets from $\cert$. In the following we argue about the relation of these sets. Let $Y_1$ and $Y_2$ be the corresponding sets from $\cert$. We will show that w.l.o.g. $Y_2 \subseteq Y_1, \{x_u,x_l\} \subseteq Y_1,x_u \in Y_2, x_l \notin Y_2$. 

First, assume for contradiction that $Y_1 \cap Y_2 = \emptyset$. Then, $\{r_{x}^{(1)},r_x^{(2)}\}\cap A = \emptyset$ since otherwise $c_A(d_x^{(1)})=0$ or $c_A(d_x^{(2)})=0$, however, both of them enter a set from $\cert$. This implies that $\{t_{x}^{(1)},t_x^{(2)}\}\cap A = \emptyset$ for all $t \in T$ such that $x \in t$ since otherwise $A$ would contain a cycle. However, no matter which incoming edge of $x_u$ is included, there exists one edge $e$ pointing from $x_l$ to $x_u$ such that $c_A(e)=0$ but $e$ enters one set from $\cert$, a contradiction. We conclude that $Y_1$ and $Y_2$ need to intersect and since $\cert$ is laminar we can assume w.l.o.g. that $Y_2 \subseteq Y_1$. Second, assume for contradiction that $x_l \notin Y_1$. By the previous argumentation we know that $x_u$ can only be entered by edges pointing from $x_l$ to $x_u$, however, these enter two sets from $\cert$, a contradiction. We conclude that $Y_2 \subseteq Y_1, \{x_u,x_l\} \subseteq Y_1,x_u \in Y_2, x_l \notin Y_2$. 

We define $\mathcal{S} := \{Y \in \cert \mid Y \text{ is $\subseteq$-maximal and there exists } Y' \in \cert, Y' \subseteq Y\}.$ Elements in $\mathcal{S}$ are non-overlapping and by the above observations we know that $|\mathcal{S}|\geq|X|$. For every $Y_1\in \mathcal{S}$ we select one representative $x(Y_1) \in X\cup Y \cup Z$ such that the node gadget of $x$ intersects with $Y_1$ and one other set from $\cert$. Considering the node gadget of $x:=x(Y_1)$, we observe that $x_u$ can only be entered by edges pointing from $x_l$ to $x_u$. We argue that in particular, $x_u$ needs to be entered by $t_x^{(1)}$ for some $t \in T$. Assume for contradiction that $x_u$ is entered by $t_x^{(2)}$  for some $t \in T$. Then, there exist $3|X|-3$ edges which are uncomparable with $t_x^{(2)}$ and hence their tails need to be included in $Y_1$. Hence, $\mathcal{S}$ contains at most $3$ sets, a contradiction to the assumption that $|X|>3$. Therefore, $x_u$ is entered by $t_x^{(1)}$ for some $t=(x,y,z)\in T$. Then, we know that $\{x_u,y_u,z_u\} \subseteq Y_1$ since $c_A((y_u,x_u))=c_A((z_u,x_u))=1$. We conclude that neither $y$ nor $z$ are included in any other set of $\mathcal{S}$. Hence, $M:=\{t \in T \mid t_{x(Y_1)}^{(1)} \in A, Y_1 \in \mathcal{S} \}$ is a $3$D-matching of size $|X|$. 
\qed
\end{proof}

\section{Branchings with low unpopularity factor}
Recall the definition of {\em unpopularity factor} from Section~\ref{sec:intro}. 
As done in the previous section, instead of studying branchings within the digraph $G$, we look at $r$-arborescences within the digraph $D$. The unpopularity factor of any branching in $G$ is the same as the unpopularity factor of the corresponding arborescence in $D$.
Given any arborescence $A$ and value $t$, there is a 
simple method to verify if $u(A) \le t$ or not. This is totally analogous to our method from Section~\ref{sec:lp} to verify popularity and it involves 
computing a min-cost arborescence in $D$ with the following edge costs. For $e=(u,v)$ in $D$, define:
 \[c_A(e) := \begin{cases} 0 \quad \ \ \ \ \ \text{if } e \succ_v A(v) \\  1\quad \ \ \ \ \ \text{if } e \sim_v A(v) \\ t+1 \quad \text{if } e \prec_{v} A(v) \end{cases}\] 

 \begin{lemma}
   \label{prop1}
   Arborescence $A$ satisfies $u(A) \le t$ if and only if $A$ is a min-cost arborescence in $D$ with edge costs given by $c_A$ defined above.
 \end{lemma}  
 \begin{proof}
For any arborescence $A'$, we now have $c_A(A') = t\cdot\phi(A,A') - \phi(A',A) + n$. We also have $c_A(A) = n$. Suppose $A$ is a min-cost arborescence 
in $D$. Then $c_A(A') \ge n$; so for any arborescence $A'$ {such that $\phi(A',A) > 0$}, we have:
\[t\cdot\phi(A,A') - \phi(A',A) \ge 0, \ \ \text{so}\ \  \frac{\phi(A',A)}{\phi(A,A')} \le t, \ \ i.e., \ \ u(A) \le t.\] 

Conversely, if $u(A) \le t$, then $t\cdot\phi(A,A') \ge \phi(A',A)$ for all arborescences $A'$. Thus $c_A(A') = t\cdot\phi(A,A') - \phi(A',A) + n \ge n$.
Since $c_A(A) = n$, $A$ is a min-cost arborescence in $D$. \qed
\end{proof}

Lemma~\ref{lem:unpop-factor} follows from Lemma~\ref{prop1} and LP-duality.

\begin{lemma}
\label{lem:unpop-factor}
Arborescence $A$ satisfies $u(A) \le  t$ if and only if there exists a dual feasible solution $y$ (see \ref{LP2} where $c_A(e)$ is as 
defined above) with $\sum_X y_X = n$.
\end{lemma}

As before, $y_X \in \{0,1\}$ for every non-empty $X \subseteq V$---thus we can identify $y$ with the corresponding set family 
${\cal F}_y = \{X \subseteq V: y_X > 0\}$. Moreover, the family ${\cal F}_y$ has at most $t+1$ levels now, i.e., if 
$X_1 \subset \cdots \subset X_k$ is a chain of sets in ${\cal F}_y$, then $k \le t+1$.

\paragraph{\bf Proof of Theorem~\ref{thm:unpop-factor}.}
We now assume node preferences are {\em strict} (thus we may assume the graph to be simple, and nodes have preferences over their in-neighbors) and  modify our algorithm from Section~\ref{sec:algo} 
so that the new algorithm always computes an arborescence $A$ in $D = (V\cup\{r\},E)$ such that $u(A) \le \lfloor\log n\rfloor$.
\begin{enumerate}
\item Initially all nodes in $V$ are active. Set $X^0_v = \{v\}$ for all $v \in V$.
\item Initialize the current edge set $E' = \emptyset$; let $i = 1$.
\item\label{build-step} Let $E' = E' \cup \{(u,v): v \in V$ is active and $u$ is $v$'s most preferred in-neighbor such that $u \notin X_v^{i-1}\}$.
\item For every active node $v$ do:
\smallskip
\begin{itemize}
\item let $X^i_v =$ set of nodes reachable from $v$ using edges in $E'$.
\end{itemize}
\smallskip
\item Let ${\cal X} = \{X^i_v \text{ is } \subseteq\text{-maximal in }\mathcal{X}^i\}$ where ${\cal X}^i = \{X^i_v: v \ \text{is active}\}$.

\ \ \ \ \ $\{${\em note that ${\cal X}^i$ is a laminar family}$\}$
\item For each $X \in {\cal X}$ do: 
\smallskip
\begin{itemize}
\item select any active node $v$ such that $X^i_v = X$; 
\item {\em deactivate} all $u \in X\setminus\{v\}$.
\ \ \ \ \ $\{${\em now $v$ is the only active node in $X$}$\}$
\item if $v$ is reachable from $r$ using edges in $E'$, then deactivate $v$.

$\{${\em this means all nodes in $X$ are reachable from $r$}$\}$
\end{itemize}
\smallskip
\item If there exists any active node, then set $i = i+1$ and go to step~\ref{build-step} above.
\item\label{last-step} Compute an arborescence $A$ in $(V \cup \{r\},E')$ and return $A$.
\end{enumerate}

We reach step~\ref{last-step} only when there is no active node. This means when we reach step~\ref{last-step}, every node is reachable from $r$ using 
the edges in $E'$. Thus there exists an arborescence $A$ in $(V \cup \{r\},E')$. Our task now is to bound its unpopularity factor.

\begin{lemma}
\label{lemma:logn}
The while-loop runs for at most $\lfloor\log n\rfloor + 1$ iterations.
\end{lemma}
\begin{proof}
Every node $v$ that is active at the start of some iteration either becomes reachable from $r$ in this iteration or it forms a weakly 
connected component that contains two or more active nodes. At the end of each iteration, there is at most one active node in each weakly 
connected component. 

So if $k$ is the number of active nodes at the start of some iteration then the number of active nodes at the end of that
iteration is at most $k/2$. Thus the number of active nodes at the end of the 
$i$-th iteration of the while-loop is at most $n/{2^i}$. Hence the while-loop can run for at most $\lfloor\log n\rfloor + 1$ iterations. \qed
\end{proof}

\begin{lemma}
\label{lemma:bound}
If the while-loop runs for $t+1$ iterations then $u(A) \le t$.
\end{lemma}
\begin{proof}
Let ${\cal Y} = \{X^{i-1}_v: v \in V \ \text{and}\ v\ \text{got deactivated in the}\ i \text{-th iteration}\}$. That is, $y_X = 1$ if 
$X \in {\cal Y}$ and $y_X = 0$ otherwise. For each node $v$, there is a corresponding set $X^{i-1}_v$ in ${\cal Y}$---note that $v$ is 
the {\em entry-point} for this set. We have $\sum_{X \subseteq V}y_X = n$. 

Since our algorithm runs for $t+1$ iterations, ${\cal Y}$ has at 
most $t+1$ levels.
For any node $v$, our algorithm ensures that the edge $(u^*,v) \in A$ is the most preferred edge entering $v$ with its tail outside 
$X^{i-1}_v$. So every other edge $e = (u,v)$ with $u \notin X^{i-1}_v$ is ranked worse than $(u^*,v) \in A$, thus $c_A(e) = t+1$. 
Hence we have $\sum_{X: \delta^-(X) \ni e}y_X \le c_A(e)$ for every edge $e$. This proves that $y$ is a feasible dual solution for $A$, so $u(A) \le t$. \qed
\end{proof}

Combining Lemmas~\ref{lemma:logn} and \ref{lemma:bound}, the first part of Theorem~\ref{thm:unpop-factor} follows. 

\paragraph{\bf A tight example.} We now describe an instance $G = (V,E)$ on $n$ nodes with strict preferences where every branching has unpopularity factor at least $\lfloor\log n\rfloor$. For convenience, let $n = 2^k$ for some integer $k$. Let $V = \{v_0,\ldots,v_{n-1}\}$.  Every node will have in-degree $k = \log n$. This instance is a generalization of the instance on 4 vertices $a, b, c, d$ given in Section~\ref{sec:intro}.
\begin{itemize}
\item For $0 \le i \le n/2-1$, the nodes $v_{2i}$ and $v_{2i+1}$ are each other's top in-neighbors. Thus $v_0,v_1$ are each other's top choice in-neighbors, $v_2,v_3$ are each other's top choice in-neighbors, and so on. 
\item The nodes $v_0,v_2$ are each other's second choice in-neighbors, similarly, $v_1,v_3$ are each other's second choice in-neighbors, and so on.
More generally, for any~$i$, if $i \in \{4j,\ldots,4j+3\}$, then the node $v_{\ell}$, where $\ell = 4j + (i+2 \bmod 4)$, is $v_i$'s second choice in-neighbor.
\item For any $i$ and any $t \in \{1,\ldots,k\}$, if $i \in \{j2^t,\ldots, (j+1)2^t - 1\}$
then the node $v_{\ell}$, where $\ell = j2^t + (i + 2^{t-1}\bmod 2^t)$, is $v_i$'s $t$-th choice in-neighbor. 
\end{itemize}

For example, $v_0$'s preference order is: $v_1 \succ v_2 \succ v_4 \succ v_8 \succ \cdots \succ v_{n/2}$.
The other preference orders are analogous. As a concrete example, let $n = 8$. So $V = \{v_0,v_1,\ldots,v_7\}$. The preference orders of all the nodes over their in-neighbors are given below.
\begin{eqnarray*}
  & v_0: \ v_1 \succ v_2 \succ v_4 \ \ \ \ \ \ \ \ \ \  & v_1: \ v_0 \succ v_3 \succ v_5 \\
  & v_2: \ v_3 \succ v_0 \succ v_6 \ \ \ \ \ \ \ \ \ \  & v_3: \ v_2 \succ v_1 \succ v_7 \\
  & v_4: \ v_5 \succ v_6 \succ v_0 \ \ \ \ \ \ \ \ \ \  & v_5: \ v_4 \succ v_7 \succ v_1 \\
  & v_6: \ v_7 \succ v_4 \succ v_2 \ \ \ \ \ \ \ \ \ \  & v_7: \ v_6 \succ v_5 \succ v_3.
\end{eqnarray*}
For any branching in the above instance (let us call it $G_k$) on $2^k$ nodes, we claim its unpopularity factor is at least $k$. We will prove this claim by induction on $k$. The base case, i.e., $k = 1$, is trivial. So let us assume that we have $u(\tilde{B}) \ge i$ for any branching $\tilde{B}$ in $G_i$. 

Consider $G_{i+1}$. Note that $v_{2j}$ and $v_{2j+1}$ are each other's top choice in-neighbors for $0 \le j \le 2^i - 1$. Let $B$ be any branching in
$G_{i+1}$. Suppose it is the case that in $B$, for some $j$: neither $v_{2j}$ is $v_{2j+1}$'s in-neighbor nor $v_{2j+1}$ is $v_{2j}$'s in-neighbor. 
Then $u(B) = \infty$, because by making $v_{2j}$ the in-neighbor of $v_{2j+1}$, no node is worse-off and $v_{2j}$ is better-off. 
We assume $u(B) < \infty$.
So it is enough to restrict our attention to the case where for each $j$ we have in $B$: 
\begin{itemize}
\item[$(\ast)$] either  $v_{2j}$ is $v_{2j+1}$'s in-neighbor or $v_{2j+1}$ is $v_{2j}$'s in-neighbor.
\end{itemize}
For each $j \in \{0,\ldots,2^i-1\}$, contract the set $\{v_{2j},v_{2j+1}\}$ into a single node in the graph $G_{i+1}$. The new graph 
(call it $G'_i$) is on $2^i$ nodes and it is exactly the same as $G_i$ except that there are 2 parallel edges between every adjacent pair of nodes 
now  -- both these edges have the same rank. 

Perform the same contraction step on the branching $B$ as well.
By $(\ast)$, it follows that the contracted $B$ (call it $B'$) is a branching such that $B'$ uses at most
1 edge in any pair of parallel edges in $G'_i$. Thus $B'$ is a branching in $G_i$ and we can use induction hypothesis to conclude that $u(B') \ge i$.

\begin{claim}
There is a branching $A'$ in  $G'_i$ such that $\phi(A',B') \ge i$ and  $\phi(B',A') = 1$. Moreover, the lone vertex that prefers $B'$ to $A'$
is a root in $A'$.
\end{claim}

We will first assume the above claim and finish our proof on $u(B)$. Then we will prove this claim. Opening up the size-2 supernodes in $B'$ will create $B$: let us run the same  ``opening up'' step on $A'$ to create a branching $A$ in $G_{i+1}$. So $\phi(A,B) \ge i$ and  $\phi(B,A) = 1$. 
We will now modify $A$ to $A^*$ so that $\phi(A^*,B) \ge i+1$ and  $\phi(B,A^*) = 1$.

Let $v_{2j}$ be the lone vertex that prefers $B$ to $A$. By the ``opening up'' step in $B$, $v_{2j+1}$ has $v_{2j}$ as its in-neighbor. The branching $A^*$ will affect only the 2 nodes $v_{2j}$ and $v_{2j+1}$ in $A$. Every other node will have the same in-neighbor in $A^*$ as in $A$. 
The above claim tells us that $v_{2j}$ is a root in $A$.
Make $v_{2j+1}$ a root in $A^*$ and $v_{2j}$'s in-neighbor will be $v_{2j+1}$. The node $v_{2j}$ was the only node that preferred $B$ to $A$ and now $v_{2j}$ prefers $A^*$ to $B$. However there is one node that prefers $B$ to $A^*$: this is $v_{2j+1}$. Recall that $v_{2j+1}$'s in-neighbor in $B$, just as in $A$, is its top-choice neighbor $v_{2j}$ while $v_{2j+1}$ is a root in $A^*$. Thus $\phi(A^*,B) \ge i+1$ and  $\phi(B,A^*) = 1$.

\paragraph{Proof of Claim.} Let $\tilde{A}$ be a branching that maximizes $\phi(\tilde{A}, B')/\phi(B', \tilde{A})$. Let $\{u_1,\ldots,u_j\}$
be the nodes that prefer $B'$ to $\tilde{A}$. There is no loss in assuming that $u_1,\ldots,u_j$ are root nodes in $\tilde{A}$. For each $i$,
let $n_i$ be the number of nodes in the arborescence rooted at $u_i$ in $\tilde{A}$ that have different in-neighbors in $\tilde{A}$ and $B'$ -- note
that each of these nodes prefers $\tilde{A}$ to $B'$ (since the ones who prefer $B'$ to $\tilde{A}$ are root nodes in $\tilde{A}$). 

Let $n_t = \max\{n_i: 1 \le i \le j\}$. 
Let $\tilde{A}_t$ be the maximal sub-arborescence of $\tilde{A}$ rooted at $u_t$, 
and let $X$ be those $n_t$ nodes in $\tilde{A}_t$ that prefer $\tilde{A}$ to $B'$.
We construct a branching $A'$. 
Let us define an arborescence $A'_t$ rooted at $u_t$ by modifying $\tilde{A}_t$ as follows:
for each $w \notin \tilde{A}_t$ that is the descendant of some $v \in \tilde{A}_t$ in $B'$, we add $B'(w)$.
We define $A'$ as the branching that contains $A'_t$ and for which $A'(v)=B'(v)$ for each $v \notin A'_t$.
So each node in $\tilde{A}_t$ has the same in-neighbor in $A'$ as in $B'$, except for the nodes in $X \cup \{u_t\}$.  

The $n_t$ nodes in $X$ prefer $A'$ to $B'$, and $u_t$ prefers $B'$ to $A'$, so we have $\frac{\phi(A', B')}{\phi(B', A')} = n_t$.
Moreover, by $u(B)<\infty$ we also have $u(B')<\infty$, which implies that every node that prefers $\tilde{A}$ to $B'$ is contained 
in a sub-arborescence of $\tilde{A}$ rooted at one of the nodes $u_1, \dots, u_j$.
Therefore we have $\phi(\tilde{A},B') = \sum_{i=1}^j n_i$, which yields
\[ \frac{\phi(A', B')}{\phi(B', A')} \ = \ n_t \ \ge \ \frac{1}{j}\sum_{i=1}^j n_i \ = \ \frac{\phi(\tilde{A}, B')}{\phi(B', \tilde{A})}.\] 
Thus the claim follows. \qed

\section{Hardness Results: Proof of Theorem~\ref{thm:hardness_load}}

\repeattheorem{restate_thm:hardness_load}
In fact, we will show that Theorem~\ref{thm:hardness_load} holds for simple graphs, therefore in this section we will assume that nodes 
have preferences over their in-neighbors. Nevertheless, we will say that an edge $(u,v)$ is a \emph{top-choice edge}, if $u$ is the best choice for $v$.
The following lemma will be useful to prove Theorem~\ref{thm:hardness_load}.

\begin{lemma}
\label{lem:topcycle}
Let $D$ be a digraph where each node has {a strict ranking} over its in-neighbors, and
let $A$ be a popular arborescence in $D$ with dual certificate~$\mathcal{Y}$.
If $C$ is a directed cycle consisting of only top-choice edges, then $A$ enters $C$ exactly once.
Let $a$ be the unique edge in $A \cap \delta^-(C)$ (guaranteed by Lemma~\ref{lem:one-to-one-sets-and-edges}), and let $Y_a$ be the unique set in $\mathcal{Y}$ entered by $a$. 
Then $C \subseteq Y_a$.
\end{lemma}

\begin{proof}
Observe that $c_A(e) \leq 1$ for any edge $e$ in $C$, as $e$ is a top-choice edge. Let $c_1,\dots, c_k$ be the nodes of $C$ in this order, 
with $a$ pointing to $c_k$.
Since $c_k$ prefers $c_{k-1}$ in $C$ to the tail of $a$, we get $c_A((c_{k-1},c_{k}))=0$ and thus $c_{k-1} \in Y_a$ by the constraints of \ref{LP2}.
Supposing $c_{k-2} \notin Y_a$ we get that $(c_{k-2},c_{k-1}) \notin A$ because exactly one edge of $A$ enters $Y_a$, by Lemma~\ref{lem:one-to-one-sets-and-edges}. 
Using that $c_{k-1}$ has strict linear preferences and prefers $c_{k-2}$ most, we obtain $c_A((c_{k-2},c_{k-1}))=0$, but this contradicts the constraints of \ref{LP2}. Hence we get $c_{k-2} \in Y_a$ as well. Repeatedly applying this argument, we get that $C \subseteq Y_a$.
\qed
\end{proof}

\medskip

\noindent{\bf Proof of Theorem~\ref{thm:hardness_load}, part~(a).}  
The reduction is from the $\mathsf{NP}$-hard problem \textsc{3-sat} where we are given a 3-CNF formula $\varphi= \bigwedge_{j=1}^m c_j$ over variables $x_1, \dots, x_n$ with each clause $c_j$ containing at most 3 literals; the task is to decide whether $\varphi$ can be satisfied. It is well known that the special case where each variable occurs at most 3 times is $\mathsf{NP}$-hard as well, so we assume this holds for $\varphi$. 

We define a digraph $D_{\varphi}$ as follows. For each variable $x_i$ we define a variable-gadget consisting of a directed 9-cycle $A_i$ on nodes $a_i^1, \dots, a_i^9$, together with nodes $t_i$ and $f_i$, both having in-degree 0 in $D_\varphi$. The top choice for any node $a_i^k$ on $A_i$ is its in-neighbor $a_i^{k-1}$ on $A_i$, its second choice is $t_i$ if $k=1$ and $f_i$ otherwise.\footnote{Throughout the rest of the proof, we treat superscripts in a circular way, that is, modulo length of the cycle in question.} 
Next, for each clause $c_j$ we define a clause-gadget as a directed cycle $C_j$ on nodes $c_j^1,\dots, c_j^h$ where $h$ is the number of literals in $c_j$; we may assume $h \in \{2,3\}$. 
The top choice for any node $c_j^k$ on $C_j$ is its in-neighbor on $C_j$. The second choice of $c_j^k$ depends on the $k$-th literal ${\ell}_j^k$ in $c_j$: it is $t_i$ if ${\ell}_j^k=x_i$, and it is $f_i$ if $\ell_j^k=\overline{x}_i$. We claim that the digraph $D_\varphi$ defined this way admits a popular branching where every node has at most 9 descendants if and only if $\varphi$ is satisfiable. 

First let us suppose that we have a satisfying truth assignment for $\varphi$; we create a branching $B$. 
If variable $x_i$ is true, then we add to $B$ the edge $(f_i,a_i^2)$ and all edges of $A_i$ except for $(a_i^1,a_i^2)$; if $x_i$ is false we add to $B$ the edge $(t_i,a_i^1)$ and all edges of $A_i$ except for $(a_i^9,a_i^1)$. For each $j \in [m]$ let us choose a literal $\ell_j^k$ in clause $c_j$ that is true according to our truth assignment. If $\ell_j^k=x_i$, then we let $B$ contain the edge $(t_i,c_j^k)$; if $\ell_j^k= \overline{x}_i$, then we let $B$ contain the edge $(f_i,c_j^k)$. 
In either case, we also add to $B$ all edges of $C_j$ but the one going into $c_j^k$; this finishes the definition of $B$. 
Observe that if $x_i$ is true, then the descendants of $f_i$ in $B$ are the nodes of $A_i$, and the descendants of $t_i$ are among the nodes of those cycles $C_j$ where $x_i$ is a literal of $C_j$; the case when $x_i$ is false is analogous. Hence, each node in $B$ has at most 9 descendants as promised. 

Let us prove that $B$ is popular. To this end, we define the graph $D'_{\varphi}$ by adding a new dummy root $r_0$ to $D_{\varphi}$ and making it the worst choice for every node in $D_{\varphi}$; moreover, we define an arborescence $A$ in $D'_{\varphi}$ by adding an edge from $r_0$ to each root of $B$. 
Then $B$ is a popular branching in $D_{\varphi}$ if and only if $A$ is a popular arborescence in $D'_{\varphi}$.
To show the latter, we define a dual certificate $\mathcal{Y}$ that contains the set $V(A_i)$ for each $i \in [n]$, the set $V(C_j)$ for each $j \in [m]$, and a singleton for each node except for those at which an edge of $B$ enters some cycle $A_i$ or $C_j$. It is straightforward to check that $\mathcal{Y}$ is indeed a dual solution proving the popularity of $A$ in $D'_{\varphi}$, and therefore of $B$ in $D_{\varphi}$.

\smallskip

Let us now suppose that we have a popular branching $B$ with each node having at most 9 descendants; we are going to define a satisfying truth assignment for $\varphi$. 
Note that the only possible roots in $B$ are the nodes in $R=\bigcup_{i \in [n]}\{ t_i,f_i\}$, since any other node $v$ has an in-neighbor in $R$ (assuming $v$ to be a root in $B$, adding an edge from $R$ to $v$ results in a branching more popular than $B$). 
Let $A$ be the popular arborescence corresponding to $B$, and let $\mathcal{Y}$ be a dual certificate proving the popularity of $A$. 

Let $e_i$ be the edge entering $A_i$ in $B$, and let $Y_{e_i}$ be the unique set in $\mathcal{Y}$ entered by $e_i$.
Similarly, let $e'_j$ be the edge entering $C_j$ in $B$, and let $Y_{e'_j}$ be the unique set in $\mathcal{Y}$ entered by $e'_j$. 
By Lemma~\ref{lem:topcycle}, we know that $A_i \subseteq Y_{e_i}$ and $C_j \subseteq Y_{e'_j}$. 

Let us define a truth assignment by setting $x_i$ true if and only if the head of $e_i$ is $f_i$. 
Note that all 9 nodes of $A_i$ are descendants of the head of $e_i$. Hence, the head of an edge $e'_j$ can only be $f_i$ if $x_i$ is false, 
and similarly, it can only be $t_i$ if $x_i$ is true. Thus, any cycle $C_j$ must be the descendant of a node representing a true literal (where $t_i$ and $f_i$ 
represent $x_i$ and $\overline{x}_i$, respectively). By the construction of $D_\varphi$, we have that any clause contains a literal set to true by the truth assignment, 
so $\varphi$ is satisfiable, proving the theorem.
\qed

\medskip 

\noindent{\bf Proof of Theorem~\ref{thm:hardness_load}, part~(b).}  
We give a reduction from the variant of the \textsc{Directed Hamiltonian Path} problem where the input digraph has a root $r$ with in-degree 0 that is the parent of all other nodes; 
it is easy to see that this version is also $\mathsf{NP}$-hard.
Let $G=(V \cup \{r\},E)$ be our given input. 
For each node we fix an arbitrary ordering on its in-neighbors, and we denote by $n(v,i)$ the $i$-th in-neighbor of a node $v \in V$.

We are going to construct a digraph $D$ that consists of a node gadget $\mathcal{G}_v$ for each $v \in V$, together with extra nodes $r$ (having in-degree 0) and $r'$. 
The gadget $\mathcal{G}_v$ consists of a \emph{core cycle} $C_v$ together with \emph{pendant cycles} $P_{v,1}, \dots, P_{v,d_v}$, each of length $d_v$, where $d_v$  denotes the in-degree of $v$ in $G$. The nodes in the core cycle are $c_v^1, \dots, c_v^{d_v}$, those in the $i$-th pendant cycle $P_{v,i}$ are $p_{v,i}^1, \dots, p_{v,i}^{d_v}$; we treat superscripts modulo $d_v$. 
The top choice for any node on these cycles is its in-neighbor within the cycle. 
The preferences are as follows, where for simplicity we define $c_r^1:=r$. 
\begin{eqnarray*}
  c_v^j &:& \ c_v^{j-1} \succ c_{n(v,j)}^1 \succ c_{n(v,j+1)}^1 \succ \dots \succ c_{n(v,d_v)}^1 \succ c_{n(v,1)}^1 \succ \dots \succ c_{n(v,j-1)}^1; \\
  p_{v,i}^j &:& \ p_{v,i-1} \succ c_v^j; \\
  r' &:& \ r. 
\end{eqnarray*}
This finishes the definition of $D$. 

We claim that $G$ has a Hamiltonian path if and only if $D$ has a popular branching with out-degree at most 2. 

For the first direction, suppose that $D$ has such a branching $B$. 
Clearly, $r$ is a root of $B$, since it has in-degree 0.
We claim that $B$ is an arborescence with root~$r$. 
First observe that any pendant cycle must be entered by $A$ once, as otherwise there exists a root of $B$ in the cycle, 
and adding the second-choice edge of this root node to $B$ (coming form a core cycle unreachable from the pendant cycle) 
we obtain a branching $B'$ that is more popular than $B$. 
Since  $(r,v) \in E$ for each $v \in V$, each node $c_v^j$ in a core cycle has an incoming edge from $r$, so such a node cannot be a root in $B$ either, proving that $B$ is indeed an arborescence. 
In particular, $\delta^-(C_v) \cap B \neq \emptyset$ for each $v \in V$.

By Lemma~\ref{lem:topcycle} we know that $B$ enters any core cycle $C_v$ exactly once, and therefore $|B \cap C_v|=d_v-1$. 
In addition, there are exactly $d_v$ edges of $B$ pointing from $C_v$ to the pendant cycles $P_{v,j}$, $j \in [d_v]$, because $B$ is an arborescence. 
This implies that there can be at most 1 edge of $B$ leaving $C_v$ and pointing to another core cycle $C_u$, as otherwise the $d_v$ nodes in $C_v$ would together have more than $2d_v$ outgoing edges in $B$, yielding that at least one of them would have out-degree 3 in $B$, a contradiction. 

Let us now define a set $H$ of edges in $G$ as follows: for each $u,v \in V$, we add $(u,v)$ to $H$ if and only if there is an edge from $C_u$ to $C_v$ in $B$.
Furthermore, we add the edge $(r,v)$ to $H$ if and only if there is an edge from $r$ to $C_v$ in $B$; 
note that there can be at most one such edge, because $(r,r') \in B$ and $B$ has out-degree at most~2. Observe that by the construction of $D$, we have $H \subseteq E$. 
Recall that by the previous paragraph, $|\delta^+(v) \cap H| \leq 1$ for each $v \in V \cup \{r\}$, and that $|\delta^-(v) \cap H| \geq 1$ for each $v \in V$. 
Moreover, $H$ is acyclic, since any cycle in $H$ would imply the existence of a cycle in $B$ as well.
Therefore, $H$ must be a Hamiltonian path. 

\smallskip

For the other direction, let $H$ be a Hamiltonian path in $G$, starting from $r$.
We define a popular branching $B$ that happens to be an arborescence.
First, for each $(u,v) \in H$ we add $(c_u^1,c_v^j)$ to $B$ where $u$ is the $j$-th in-neighbor of $v$, 
and we also add all edges of $C_v$ to $B$ except for the one pointing to $c_v^j$;
note that here we cover the case where $u=r$ as well. 
Notice that for each $v \in V$ there are at most $d_v$ edges of $B$ whose tail is in $C_v$.
Hence, there exist $d_v$ edges in $\delta^+(C_v)$ whose addition to $B$ does not violate our bound on the out-degree 
and such that each of these edges points to a distinct pendant cycle $P_{v,j}$ (note that any pendant cycle can be connected to any node on $C_v$).
Let us add these edges to~$B$ as well, together with all edges in $P_{v,j}$ except for the one whose head already has an incoming edge in~$B$, 
for each $j \in [d_v]$. Finally, we add the edge $(r,r')$ to $B$. 
It is easy to verify that the edge set $B$ obtained this way is indeed an arborescence with root~$r$, 
and has out-degree at most 2. 

It remains to show that $B$ is popular.
To this end, we define the graph $D'$ by adding a new dummy root $r_0$ to $D$ and making it the worst choice for every node in $D$; 
moreover, we define an arborescence $A$ in $D'$ by adding the edge $(r_0,r)$ to~$B$. 
Then $B$ is a popular branching in $D$ if and only if $A$ is a popular arborescence in $D'$.
To prove the latter, we provide a dual certificate $\mathcal{Y}$ as follows. 
For each core or pendant cycle $C$, we put the set $V(C)$ into $\mathcal{Y}$, together with a singleton $\{v\}$ for each $v \in V(C)$ 
except for the node at which $B$ enters $C$. We also add singletons $\{r'\}$ and $\{r\}$.
The set system $\mathcal{Y}$ so obtained contains exactly $|V(D)|$ sets, so it remains to show that it fulfills the conditions of
\ref{LP2}. First note that any edge may enter at most two sets from $\mathcal{Y}$. 
Note also that if $v$ is a node such that $B$ enters a core or pendant cycle $C$ at $v$, then $\delta^-(v) \cap B$ is the second choice for $v$ 
(and its best choice is within $C$, the set of $\mathcal{Y}$ corresponding to $v$);
otherwise $\delta^-(v) \cap B$ is the best choice for $v$. From these facts it is straightforward to verify the constraints of \ref{LP2}, 
so the popularity of $B$ and hence the theorem follows.
\qed

\section{Popular mixed branchings}
A \emph{mixed branching} $P$ is a probability distribution (or lottery) over branchings in $G$,
i.e., $P = \{(B_1,p_1)\ldots,(B_k,p_k)\}$, where $B_i$ is a branching in $G$ for each~$i$ and
$\sum_{i=1}^k p_i = 1$, $p_i \ge 0$ for all $i$. Popular mixed {\em matchings} were studied in \cite{KMN09} where it was
shown that popular mixed matchings always exist and can be efficiently computed. Using the proof and method in \cite{KMN09}, we now show
that popular mixed branchings also always exist and such a mixed branching can be computed in polynomial time.

The function $\phi(B,B')$ that allowed us to compare two branchings $B, B'$ generalizes to mixed
branchings in a natural way. For mixed branchings $P = \{(B_1,p_1)\ldots,(B_k ,p_k)\}$ and
$Q = \{(B'_1,q_1)\ldots,(B'_{\ell},q_{\ell})\}$, the function $\phi(P,Q)$ is the expected number of nodes that prefer
$B$ to $B'$ where $B$ and $B'$ are drawn from the probability distributions $P$ and $Q$ respectively; in other words,
\[   \phi(P,Q) = \sum_{i=1}^k \sum_{j=1}^l\ p_i\, q_j\, \phi(B_i, B'_j).\]

\begin{definition}
  \label{def:mix-pop-branching}
  A mixed branching $P$ is popular if $\phi({P},{Q}) \geq \phi({Q},{P})$ for all mixed branchings~$Q$. 
\end{definition}

Consider the instance on 4 nodes $a, b, c, d$ described in Section~\ref{sec:intro} that did not admit any
popular branching. Let $B_1 = \{(a,b),(b,d),(d,c)\}$, $B_2 = \{(b,a),(a,c)$, $(c,d)\}$, $B_3 = \{(c,d),(d,b),(b,a)\}$, and
$B_4 = \{(d,c),(c,a),(a,b)\}$. It can be verified that the mixed matching $P = \{(B_1,1/4),(B_2,1/4),(B_3,1/4)$, $(B_4,1/4)\}$
is popular.

\begin{proposition}
\label{prop:mixed}
 Every instance $G$ admits a popular mixed branching.
 \end{proposition}

The proof of the above proposition is the same as the one given in \cite{KMN09} for popular mixed matchings.
Consider a two-player zero-sum game where the rows and columns of the payoff matrix $M$ are indexed by all branchings
$B_1,\ldots,B_N$ in $G$. The $(i,j)$-th entry of the matrix $M$ is  $\Delta(B_i,B_j) = \phi(B_i,B_j) - \phi(B_j,B_i)$.
A mixed strategy of the row player is a probability distribution $\langle p_1, \ldots, p_N\rangle$ over the rows
of~$M$; similarly, a mixed strategy of the column player is a probability distribution
$\langle q_1, \ldots, q_N\rangle$ over the columns of~$M$.

The row player seeks to find a mixed branching $P$ that maximizes $\min_Q\Delta(P,Q)$. The column player
seeks to find a mixed branching $Q$ that minimizes $\max_P\Delta(P,Q)$. We have:
\[0 \ \le \ \min_Q\max_P\Delta(P,Q) \ = \ \max_P\min_Q\Delta(P,Q) \ \le \ 0,\]
where the first inequality follows by taking $P = Q$, the last inequality follows by taking $Q = P$, and
the (middle) equality follows from Von~Neumann's minimax theorem. Thus $\max_P\min_Q\Delta(P,Q) = 0 $, i.e.,
there exists a probability distribution $P$ over branchings such that $\Delta(P,Q) \ge 0$ for all mixed branchings $Q$. In other words,
 $P$ is a popular mixed branching. 
 
 \paragraph{\bf Computing a popular mixed branching.} {Since branchings in $G$ and $r$-arborescences in $D = (V\cup\{r\},E)$ are equivalent with respect to popularity}, we will work 
 in the graph $D$ now. Analogous to \cite{KMN09}, instead of {\em mixed} arborescences, it will be more convenient to deal with {\em fractional} arborescences. 
 
 A fractional arborescence $x$ is a point in the arborescence polytope ${\cal A}$ of $D$, i.e., $x$ is a point that satisfies constraints~\eqref{constr1}-\eqref{constr2}. So $x$ is a convex combination of arborescences in $D$, i.e., it is a mixed arborescence
 $\{(A_1,\alpha_1)\ldots,(A_k ,\alpha_k)\}$ where $x = \sum_j \alpha_jI_{A_j}$ (note that there may be multiple ways of expressing $x$ as a mixed
 arborescence). 
 
 Conversely, every mixed arborescence $P = \{(A'_1,p_1)\ldots,(A'_t ,p_t)\}$ maps to a fractional arborescence $\sum_k p_kI_{A'_k}$, where $I_{A'_k}$ is the incidence vector of arborescence $A'_k$. Thus there is a many-to-one mapping between mixed arborescences and fractional arborescences. Given a 
 fractional arborescence $x$, we can efficiently find an equivalent mixed arborescence whose support is at most $m$ using Carath\'{e}odory's theorem.
 
 For any two fractional arborescences $x, y$, define $\Delta(x,y)$ as follows:
 \[\Delta(x,y) = \sum_{u\in V}\sum_{\substack{e\in\delta^-(u)\\e'\in\delta^-(u)}} x_e\, y_{e'}\,\vote_u(e,e'),\] 
 where $\vote_u(e,e')$ is 1, 0, -1 corresponding to $e \succ_u e'$, $e \sim_u e'$, and $e \prec_u e'$, respectively. Let $P, Q$ be
 two mixed arborescences and let $x,y$ be their corresponding fractional arborescences. It is easy to show that $\Delta(P,Q) = \Delta(x,y)$.
 
 A popular fractional arborescence $x$ is popular if $\Delta(x,y) \ge 0$ for all fractional arborescences $y$. It follows from Proposition~\ref{prop:mixed}
 that popular fractional arborescences always exist in $D$. The following linear program finds a popular fractional arborescence $x$.
\begin{linearprogram}
  {
    \label{mixed-LP} 
    \minimize 0
  }
  \Delta(x,A) & \ \geq \ 0 \ \ \  \forall\, \text{arborescences}\ A \ \text{in}\ D\notag \\
  x & \ \in\ {\cal A}  \notag
\end{linearprogram}

The feasible region of \ref{mixed-LP} is the set of fractional arborescences that do not lose to any {\em integral} arborescence. This immediately implies
that such a fractional arborescence is a popular fractional arborescence. 

There are 2 sets of exponentially many constraints in \ref{mixed-LP}. Both sets of constraints admit efficient separation oracles: to decide if $x \in {\cal A}$ or not, a min $r$-cut needs to be computed in $D$ with edge capacities given by $x$. If this cut $(S\cup\{r\}, V\setminus S)$ has value less
than 1, then the set $V \setminus S$ forms a violating constraint w.r.t. \eqref{constr1}; else $x \in {\cal A}$.

To decide if $\Delta(x,A) \ge 0$ for all arborescences $A$, we compute a min-cost arborescence in $D$ with the following edge costs:
\[ c_x(e) \ = \ \sum_{e' \succ_u e} x_{e'} - \sum_{e' \prec_u e} x_{e'} \ \ \ \forall e \in E.\]
It is simple to check that for any arborescence $A$, we have $c_x(A) = \Delta(x,A)$. Thus $x$ is {\em unpopular} if and only if there is an arborescence
$A$ with $c_x(A) < 0$. 

Since a min-cost arborescence can be computed in polynomial time~\cite{Edmo67a,KoVy06a}, we can efficiently find a violating
constraint $\Delta(x,A) < 0$ if $x$ is unpopular. Thus we can compute a popular mixed arborescence in polynomial time using the ellipsoid method.
Hence we have shown the following theorem.

\begin{theorem}
\label{thm:mixed}
A popular mixed branching in a digraph $G$ where every node has preferences in arbitrary partial order over its incoming edges can be computed in polynomial time.
\end{theorem}

\end{document}